\documentclass[conference]{IEEEtran}
\IEEEoverridecommandlockouts
\usepackage{cite}
\usepackage[letterpaper, top=0.7in, bottom=1in, left=0.625in, right=0.625in]{geometry}
\usepackage{amsmath,amssymb,amsfonts}
\usepackage{amsthm} 
\newtheorem{proposition}{Proposition}
\usepackage{booktabs}
\usepackage{array}
\usepackage{algorithmic,algorithm}
\newtheorem{assumption}{Assumption}

\usepackage{caption,subcaption}
\usepackage{graphicx}
\usepackage{placeins}
\usepackage{csquotes}
\usepackage{textcomp}
\usepackage{xcolor}
\newcommand{\rev}[1]{#1}
\newcommand{\revcite}[1]{\cite{#1}}

\def\BibTeX{{\rm B\kern-.05em{\sc i\kern-.025em b}\kern-.08em
    T\kern-.1667em\lower.7ex\hbox{E}\kern-.125emX}}
\begin{document}
\bstctlcite{BSTcontrol}
\title{Grey-Box Bayesian Optimization for ISAC in Fluid-Antenna Assisted Air-Ground Network % grey-box
}

\author{\IEEEauthorblockN{
    Gangyong Zhu, \textit{Graduate Student Member, IEEE},
    Jia Yan, \textit{Member, IEEE},
    Miaowen Wen, \textit{Senior Member, IEEE}, \\
    and Shijian Gao, \textit{Member, IEEE}
}%
% \thanks{Part of this paper has been submitted to the 2026 IEEE Global Communications Conference (GLOBECOM).}
}

% 190 words limitation
\maketitle
\begin{abstract}

Fluid antenna systems (FAS) provide additional spatial diversity for integrated sensing and communication (ISAC) through joint port selection and precoding. \rev{However, existing designs commonly assume readily available channel state information, neglect residual self-interference, and combine communication and sensing into a single weighted objective. The resulting channel acquisition overhead, unmodeled residual self-interference, and limited characterization of the Pareto trade-off are critical obstacles to implementing ISAC in fluid-antenna-assisted air-ground networks.} \rev{To address these issues, we first formulate the joint design as a grey-box multi-objective optimization problem. This formulation retains the available analytical system mapping while treating the configuration-dependent channel and interference constituents as unknown, and directly represents the communication-sensing Pareto trade-off without predefined scalarization.} We then propose a tailored grey-box multi-objective Bayesian optimization (G-MOBO) method to solve the resulting high-dimensional problem. Specifically, G-MOBO learns the unknown constituents from performance feedback, propagates their predictive distributions through the known mapping, employs expected hypervolume improvement (EHI) to explore the Pareto frontier, and uses an adaptive trust region (TR) to localize the search. A temporal adaptation strategy is further incorporated to track the drifting Pareto frontier in time-varying environments. \rev{The theoretical analysis characterizes the sample efficiency of grey-box modeling and localized TR design via cumulative hypervolume regret.} Simulations demonstrate faster convergence, improved Pareto-frontier quality, and robust dynamic tracking compared with the considered baselines.
\end{abstract}


\section{Introduction}

\IEEEPARstart{I}{n} future 6G dynamic air-ground networks, there is a critical demand for simultaneous high-speed data transmission to ground users and precise sensing of aerial targets\cite{gao2026iscc_lan}.  To meet this dual requirement under tight spectrum and hardware constraints, integrated sensing and communication (ISAC) has emerged as a pivotal technology. However, conventional ISAC systems typically rely on fixed-position antenna arrays, whose spatial degrees of freedom (DoF) fundamentally limit their ability to adapt to the rapidly fluctuating channels and complex interference inherent in highly dynamic air-ground environments. To overcome these hardware bottlenecks, fluid antenna systems (FAS) have recently garnered extensive attention as a promising solution. Distinct from traditional fixed arrays, FAS leverages the extended spatial diversity within a predefined aperture by finely reconfiguring the positions of radiating elements \cite{new2024tutorial,zhu2024movable}. This paradigm introduces position-agile spatial DoF, offering the system unprecedented flexibility to mitigate interference and optimize the delicate trade-off between sensing and communication performance \cite{cheng2024intelligent,new2025fluid}.

Fully realizing the potential of FAS-ISAC requires the joint optimization of port selection and precoding \cite{zhang2025efficient,xiu2025movable,lyu2025movable}. Conventional analytical (white-box) methods have laid a solid foundation for this task by relying on explicit and accurate channel state information (CSI) \cite{zhao2024performance,gao2020doublysparse}. While effective under well controlled conditions, these approaches face practical bottlenecks in air-ground networks. The double dynamics stemming from user and target mobility, combined with fluid antenna switching, significantly shorten the channel coherence time and complicate accurate CSI estimation\cite{yang2025synesthesia}. Furthermore, the non-negligible near-field coupling and configuration-dependent non-linear self-interference (SI) inherent in FAS introduce complex hardware characteristics that are difficult to model analytically \cite{hong2025fluid}. Even robust designs accommodating imperfect CSI often experience performance degradation in practice \revcite{hao2025fluid}. \rev{Recent movable antenna (MA)-enabled ISAC designs also jointly optimize antenna positions, communication beamformers, and sensing covariance under dynamic radar cross section (RCS) uncertainty \revcite{khalili2024advanced}. This formulation has been extended to a two-timescale design that accounts for slow antenna movement, dynamic RCS, and imperfect CSI \revcite{khalili2026movable}.} This fundamental limitation of explicit analytical modeling necessitates a shift towards alternative strategies that bypass complex channel formulations, aiming instead to learn high-quality configurations directly from end-to-end performance feedback via black-box optimization.

% Preserve the established reference numbering while declaring Table~\ref{tab:related_work}
% early enough for it to appear at the top of the second page.
\nocite{song2023two,zeng2024csi,cheng2025bayesian,yuan2025optimal,frazier2018tutorial,eriksson2019scalable,liu2022learning,li2025ekf,zhang2024predictive,liu2020radar,yan2023bayesian}

{\color{black}
\begin{table*}[t]
\centering
\caption{Feature comparison with the five most closely related works.}
\label{tab:related_work}
\vspace{5pt}
\setlength{\tabcolsep}{6pt}
\renewcommand{\arraystretch}{1.22}
\footnotesize
\begin{tabular*}{\textwidth}{@{\extracolsep{\fill}}lcccc@{}}
\toprule
\textbf{Ref.} & \shortstack{\textbf{CSI}\\\textbf{availability}} & \shortstack{\textbf{Multi-objective}\\\textbf{formulation}} & \shortstack{\textbf{Residual}\\\textbf{self-interference}} & \shortstack{\textbf{Online}\\\textbf{adaptation}} \\
\midrule
\revcite{hao2025fluid} & $\triangle$ & $\triangle$ & $\times$ & $\times$ \\
\revcite{khalili2024advanced} & $\times$ & $\triangle$ & $\times$ & $\times$ \\
\revcite{khalili2026movable} & $\triangle$ & $\triangle$ & $\times$ & $\checkmark$ \\
\revcite{zeng2024csi} & $\checkmark$ & $\times$ & $\times$ & $\times$ \\
\revcite{cheng2025bayesian} & $\checkmark$ & $\times$ & $\times$ & $\checkmark$ \\
\midrule
\textbf{This work} & $\checkmark$ & $\checkmark$ & $\checkmark$ & $\checkmark$ \\
\bottomrule
\end{tabular*}\par
\vspace{1.5mm}
\parbox{0.98\textwidth}{\footnotesize \textit{Legend:} $\checkmark$: explicitly considered; $\triangle$: partially or indirectly considered; $\times$: not explicitly considered.}
\end{table*}
}

To execute the black-box optimization, existing literature has explored diverse strategies, ranging from multi-armed bandit (MAB) \cite{song2023two} and zeroth-order optimization (ZO) \revcite{zeng2024csi} to data-efficient Bayesian optimization (BO) \revcite{cheng2025bayesian}. For instance, MAB methods treat the configuration selection as a sequential decision-making process \cite{song2023two}. However, standard MAB methods typically treat each discrete system configuration as an independent arm. By neglecting the inherent spatial smoothness of the FAS aperture, they fail to transfer knowledge across neighboring ports, resulting in a brute-force exploration process that scales poorly with the number of antennas. Alternatively, ZO algorithms attempt to estimate the gradient of the objective function through random perturbations \revcite{zeng2024csi}. Yet, in the high-dimensional discrete search space of FAS-ISAC, this gradient estimation becomes computationally imprecise \cite{yuan2025optimal}. To improve sample efficiency, recent works have applied BO by treating the end-to-end objective function as a completely unknown black-box \revcite{cheng2025bayesian}, \cite{frazier2018tutorial}. \rev{Conventional Gaussian process (GP)-based BO has been investigated for antenna-position training in movable-antenna communication systems \revcite{cheng2025bayesian}. However, this method adopts a single-objective black-box formulation and does not address multi-objective ISAC design.}

Furthermore, the aforementioned approaches are predominantly developed under a quasi-static assumption, rendering them susceptible to performance degradation as the environment evolves. In practical air-ground deployments, the continuous motion of unmanned aerial vehicles (UAVs) and targets induces a strictly time-varying environment where the optimal configuration drifts continuously \cite{liu2022learning}. If the optimizer fails to adapt rapidly, the selected configuration becomes stale shortly after application, leading to severe beam misalignment and compromised ISAC utility. An intuitive remedy is to employ predictive beamforming techniques, such as the extended Kalman filter (EKF) \cite{li2025ekf} or long short-term memory (LSTM) networks \cite{zhang2024predictive}, to forecast the channel state and subsequently re-optimize. These predict-then-optimize methods are effective when the tracked state and the configuration model are sufficiently accurate. In the considered FAS-ISAC setting, however, configuration-dependent SI coupling and hardware mismatch may remain difficult to capture even when the physical propagation channel is accurately tracked. Consequently, minimizing channel-prediction error alone does not necessarily maximize the end-to-end ISAC utility \cite{liu2020radar,yan2023bayesian}.

% These critical considerations strongly motivate the development of a sample-efficient framework explicitly designed to track a drifting optimum directly through end-to-end performance feedback.

To overcome the aforementioned challenges, we propose a tailored grey-box multi-objective Bayesian optimization (G-MOBO) method, which is further extended to handle strictly time-varying environments. Specifically, we first formulate the FAS-ISAC joint design as a multi-objective optimization problem to fundamentally capture the communication-sensing trade-off without relying on predefined scalarization. To efficiently solve this, we exploit the system's underlying grey-box nature, explicitly leveraging known analytical models, thereby reducing sample complexity. Building upon this surrogate modeling, the expected hypervolume improvement (EHI) acquisition function is employed to intelligently balance exploration and exploitation across the dual-objective Pareto frontier. Furthermore, to combat the severe over-exploration inherent in the high-dimensional combinatorial discrete search space, an adaptive trust region (TR) mechanism is seamlessly integrated to restrict candidate generation within localized hyper-regions. The main contributions of this work are summarized as follows:
\begin{itemize}
    \item \rev{We formulate the joint port-selection, array-orientation, and precoding design as a direct multi-objective optimization problem. Unlike weighted single-objective formulations, it explicitly represents the communication-sensing Pareto trade-off without requiring instantaneous CSI and while accounting for configuration-dependent residual SI.}

    \item \rev{We propose the G-MOBO framework based on grey-box constituent modeling. Instead of learning the final objectives as unrelated black boxes, G-MOBO learns the configuration-dependent communication channel, target response, clutter, and residual-SI constituents and propagates their uncertainty through the known ISAC performance mapping. A GP provides the theoretical reference realization, while a random forest (RF) offers a practical discrete-aware realization.}
    
    \item \rev{Within G-MOBO, we develop a Pareto-aware TR acquisition strategy for efficient search in the high-dimensional discrete configuration space. EHI prioritizes candidates according to their expected contribution to the Pareto frontier, while adaptive TR resizing localizes candidate generation and is governed by the achieved hypervolume improvement.}
    
    % 4. 鎷撳睍锛氬叿澶噒ime-varying
    \item We extend the proposed framework to dynamic air-ground scenarios by introducing a lightweight temporal adaptation strategy. This extension tracks the drifting multi-objective optimum directly from recent performance feedback without requiring an explicit channel-prediction and configuration model.
\end{itemize}

The remainder of this paper is organized as follows. Section II introduces the system model, detailing the channel characteristics and the mobility models for dynamic scenarios. Section III presents the proposed optimization framework, first detailing the grey-box multi-objective BO algorithm for static scenarios and then extending it to a framework for dynamic environments. Section IV provides a theoretical analysis, including convergence properties and computational complexity. Section V presents the simulation results and performance analysis. Finally, Section VI concludes the paper.

\textit{Notation}: In this paper, scalars are denoted by italic letters (e.g., $a$), vectors by bold lowercase letters (e.g., $\mathbf{a}$), and matrices by bold uppercase letters (e.g., $\mathbf{A}$). $\mathbb{C}^{M \times N}$ and $\mathbb{R}^{M \times N}$ denote the spaces of $M \times N$ complex-valued and real-valued matrices, respectively. The superscripts $(\cdot)^T$ and $(\cdot)^H$ represent the transpose and conjugate transpose operations, respectively. $\operatorname{Tr}(\cdot)$, $\det(\cdot)$, and $\operatorname{diag}(\cdot)$ denote the trace, determinant, and diagonalization operators. $\|\mathbf{a}\|$ represents the Euclidean norm ($L_2$-norm) of vector $\mathbf{a}$, while $|\mathcal{A}|$ denotes the cardinality of set $\mathcal{A}$. The expectation of a random variable is denoted by $\mathbb{E}[\cdot]$. $\mathcal{N}(\boldsymbol{\mu}, \boldsymbol{\Sigma})$ represents the multivariate Gaussian distribution with mean $\boldsymbol{\mu}$ and covariance matrix $\boldsymbol{\Sigma}$. Finally, $\mathbb{I}(\cdot)$ denotes the indicator function, which equals 1 if the condition is satisfied and 0 otherwise.

% fig.1 - fig.2 鍚堟垚涓轰竴寮?
% 鐢ㄦ埛鍜岀洰鏍囬渶瑕佽娓呮婊¤冻鐨勬潯浠躲€?
\section{System Model}
In this section, we begin by detailing the FAS-ISAC setup, followed by establishing the channel model for communication and sensing. Subsequently, we extend the channel model to the temporal domain by introducing the target and user mobility.
\subsection{FAS-enabled ISAC Networks}
Consider a millimeter-wave (mmWave) ISAC system illustrated in Fig.~\ref{fig:system-model}, where a base station (BS) simultaneously serves $K$ single-antenna users and senses $Q$ aerial targets using a planar FAS array. The array is deployed over a rectangular region $\mathcal{C}$ of size $A\lambda \times A\lambda$, which contains a dense and uniformly distributed grid of candidate ports with adjacent spacing $d$. From these candidate ports, the BS activates $N=N_t+N_r$ fluid antennas, where $N_t$ antennas are assigned to signal transmission and $N_r$ antennas to echo reception, subject to the DoF constraints $N_t \ge K$ and $N_r \ge Q$, respectively. In addition to port selection, the entire array plane can be mechanically rotated with azimuth angle $\varphi$ and elevation angle $\theta$. Accordingly, the effective 3D global position of the $m$-th active antenna is jointly determined by its selected port and the array orientation, and is denoted by $\mathbf{p}_m(\theta,\varphi)=[x_m,y_m,z_m]^T$. Let $\mathbf{f}_k \in \mathbb{C}^{N_t \times 1}$ denote the beamforming vector for the $k$-th user. The BS transmits a dual-functional signal $\mathbf{x} \in \mathbb{C}^{N_t \times 1}$ that superimposes the communication symbols for all users:
\begin{equation}
\mathbf{x} = \sum_{k=1}^{K} \mathbf{f}_k s_k = \mathbf{F}\mathbf{s},
\end{equation}
where $\mathbf{s} = [s_1, \dots, s_K]^T \in \mathbb{C}^{K \times 1}$ represents the independent data symbols vector, which is normalized by $\mathbb{E}[\mathbf{s}\mathbf{s}^H] = \mathbf{I}_K$. The transmit power is constrained by $\mathbb{E}[\|\mathbf{x}\|^2] = \operatorname{Tr}(\mathbf{F}\mathbf{F}^H) \leq P_t$.

\begin{figure}[t]
  \centering
  \includegraphics[width=0.9\linewidth]{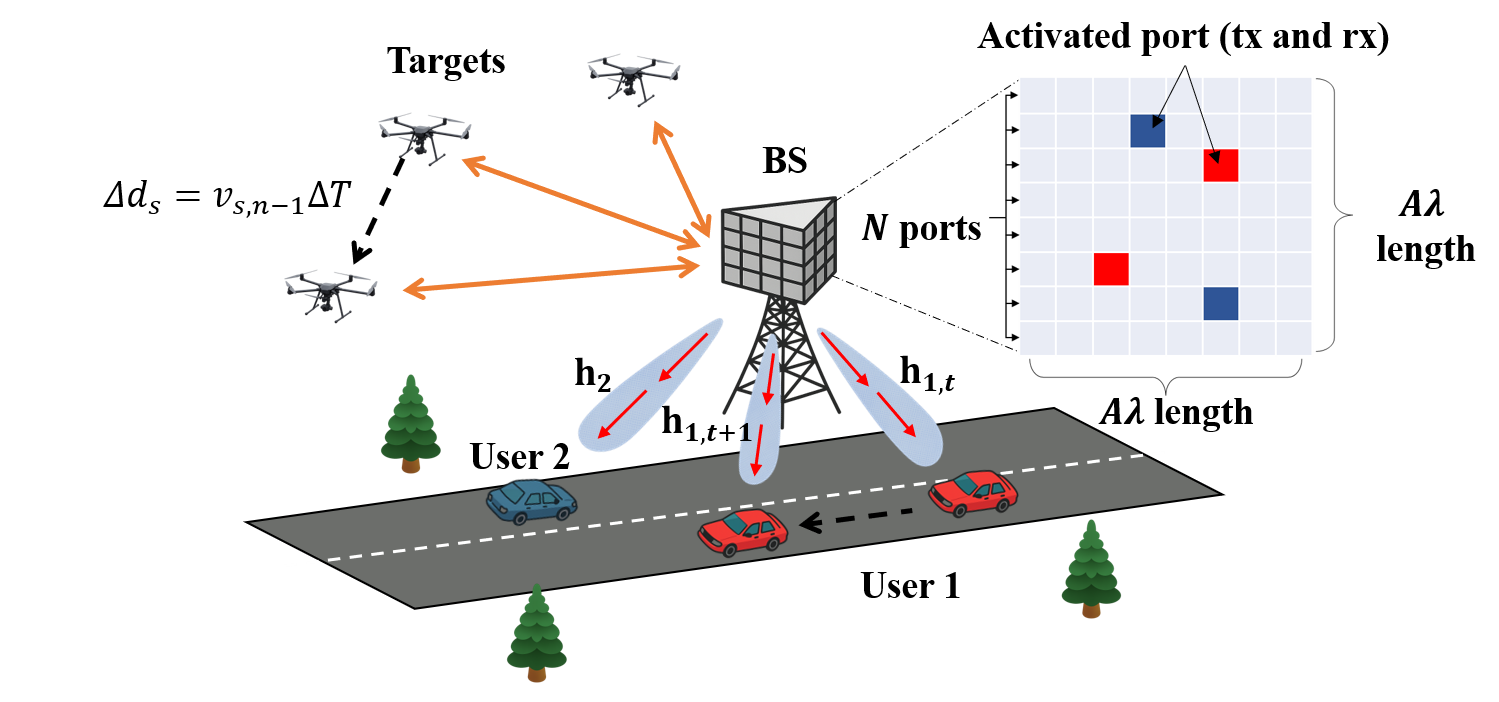}%
  \caption{System of the proposed FAS-ISAC setup.}
    \label{fig:system-model}
\end{figure}
% 绮剧畝鍑犺 (4-5) 娈佃惤杩囧(2娈佃冻澶?
% input-output relationship
\subsection{Communication and Sensing Model}
We first establish the channel model for the downlink communication. \rev{Let ${\mathbf{P}}_t = [{\mathbf{p}}_{t,1}, \dots, {\mathbf{p}}_{t,N_t}] \in \mathbb{R}^{3 \times N_t}$ denote the position matrix of the $N_t$ active transmit antennas.} The propagation environment is characterized by a multipath channel model. For the $k$-th user, the channel consists of $L_k$ distinct propagation paths, each associated with a complex path gain $\sigma_{k,l}$ and a specific spatial direction. The phase response of the $m$-th transmit antenna relative to the coordinate origin is determined by the signal propagation delay, which corresponds to the projection of the antenna position onto the wave vector. This phase shift function of the $m$-th port is given by $\rho(\mathbf{{p}}_{t,m}, \theta_{k,l}, \varphi_{k,l}) = \mathbf{{p}}^T_{t,m}\mathbf{u}_{k,l}(\theta_{k,l}, \varphi_{k,l})$, where $\mathbf{u}_{k,l}$ is the normalized direction vector of the $l$-th path for the $k$-th user. $\theta_{k,l}$ and $\varphi_{k,l}$ represent the elevation and azimuth angles-of-departure (AoD) of the $l$-th path, respectively.

Based on the fluid antenna architecture, the channel vector is constructed using the field response matrix (FRM). We define the FRM for the $k$-th user, denoted as $\mathbf{G}_k({\mathbf{P}}_t) \in \mathbb{C}^{L_k \times N_t}$, which captures the phase shifts of all multipaths across the selected antenna ports. \rev{The $(l, m)$-th entry of $\mathbf{G}_k({\mathbf{P}}_t)$ represents the steering vector element for the $l$-th path at the $m$-th active antenna position}\cite{zhu2024modeling}:
\begin{equation}
\rev{[\mathbf{G}_k({\mathbf{P}}_t)]_{l,m} = e^{-j\frac{2\pi}{\lambda}\rho(\mathbf{p}_{t,m},\theta_{k,l},\varphi_{k,l})}.}
\end{equation}
Consequently, the aggregated communication channel vector $\mathbf{h}_k({\mathbf{P}}_t) \in \mathbb{C}^{N_t \times 1}$ is expressed as the superposition of these path responses weighted by their respective gains \cite{zhu2024multiuser}:
\begin{equation}
\mathbf{h}_k({\mathbf{P}}_t) = \mathbf{G}^T_k({\mathbf{P}}_t) \boldsymbol{\sigma}_k,
\end{equation}
where $\boldsymbol{\sigma}_k = [\sigma_{k,1}, \dots, \sigma_{k,L_k}]^T \in \mathbb{C}^{L_k \times 1}$ is the vector of complex path coefficients. The communication performance for the $k$-th user is evaluated via the Signal-to-Interference-plus-Noise Ratio (SINR):
\begin{equation}
\gamma_k = \frac{|\mathbf{h}^T_k({\mathbf{P}}_t)\mathbf{f}_k|^2}{\sum_{q \neq k}|\mathbf{h}^T_k({\mathbf{P}}_t)\mathbf{f}_q|^2 + \sigma_c^2},
\end{equation}
where $\sigma_c^2$ denotes the additive white Gaussian noise (AWGN) variance. The communication utility is typically defined as the sum rate $R_c = \sum_{k=1}^K \log_2(1 + \gamma_k)$.

Simultaneously, the BS utilizes the transmitted signal $\mathbf{x}$ to detect $Q$ aerial targets. The echo signals are captured by $N_r$ active receive antennas. Similar to the transmit side, let ${\mathbf{P}}_r = [{\mathbf{p}}_{r,1}, \dots, {\mathbf{p}}_{r,N_r}] \in \mathbb{R}^{3 \times N_r}$ denote the position matrix of the active receive ports, where ${\mathbf{p}}_{r,n}$ represents the 3D coordinate of the $n$-th receive antenna. In this monostatic sensing configuration, the received signal consists of the target echoes $\mathbf{G}_{{t}}$, the environmental clutter $\mathbf{G}_{{clu}}$, and the residual SI leakage from the transmitter $\mathbf{G}_{{si}}$. $\mathbf{G}_{{t}}(\cdot)$ captures the far-field reflection from the $Q$ point targets. Utilizing the phase response function defined in (3), it is expressed as
\begin{equation}
\mathbf{G}_{{t}}({\mathbf{P}}_r, {\mathbf{P}}_t) = \sum_{q=1}^Q \alpha_q \mathbf{b}({\mathbf{P}}_r, \theta_q, \varphi_q) \mathbf{a}^H({\mathbf{P}}_t, \theta_q, \varphi_q),
\end{equation}
where $\alpha_q$ is the complex reflection coefficient of the $q$-th target. The vectors $\mathbf{a}({\mathbf{P}}_t, \cdot) \in \mathbb{C}^{N_t \times 1}$ and $\mathbf{b}({\mathbf{P}}_r, \cdot) \in \mathbb{C}^{N_r \times 1}$ are the transmit and receive steering vectors, respectively, constructed by stacking the phase terms $e^{-j\frac{2\pi}{\lambda}\rho(\cdot)}$ for the selected ports. \rev{The environmental clutter is modeled as $Q_c$ static point scatterers. For the $c$-th scatterer, let $(\theta_c^{\rm clu},\varphi_c^{\rm clu})$ and $\alpha_c^{\rm clu}$ denote its angle pair and complex reflection coefficient, respectively. The aggregate clutter channel is}
\begin{equation}
\rev{\mathbf G_{\rm clu}(\mathbf P_r,\mathbf P_t)=\sum_{c=1}^{Q_c}\alpha_c^{\rm clu}\mathbf b(\mathbf P_r,\theta_c^{\rm clu},\varphi_c^{\rm clu})\mathbf a^H(\mathbf P_t,\theta_c^{\rm clu},\varphi_c^{\rm clu}).}
\end{equation}
Unlike the far-field targets, SI leakage arises from direct coupling between transmit and receive antennas on the same aperture, where the inter-port distance is comparable to the wavelength and the plane-wave assumption breaks down. \rev{The nominal spherical-wave component is modeled as $[\mathbf{G}_{{si}}^{\rm nom}]_{n,m} = \frac{\beta}{d_{n,m}} e^{-j\frac{2\pi}{\lambda} d_{n,m}}$, where $\beta$ is an SI-cancellation-dependent attenuation coefficient and $d_{n,m} = \|\mathbf{p}_{r,n} - \mathbf{p}_{t,m}\|_2$. In the simulations, this component forms the line-of-sight term of a Rician residual-SI model with $K_{\rm SI}=10$~dB, while a calibrated static scattering term represents nonideal coupling. The resulting configuration-dependent effective matrix $\mathbf G_{\rm si}$ is treated as a learned constituent in the grey-box model.} \revcite{hao2025fluid} The effective sensing channel matrix $\mathbf{G}_s({\mathbf{P}}_r, {\mathbf{P}}_t) \in \mathbb{C}^{N_r \times N_t}$ describes the total propagation from the transmit ports ${\mathbf{P}}_t$ to the receive ports ${\mathbf{P}}_r$, and can be decomposed as
\begin{equation}
\mathbf{G}_{{s}}(\cdot) = \mathbf{G}_{{t}}(\cdot) + \mathbf{G}_{{clu}}(\cdot) + \mathbf{G}_{{si}}(\cdot).
\end{equation}
Hence, the received signal at the sensing sub-array is given by:
\begin{equation}
\mathbf{Y} = \mathbf{G}_{{s}} \mathbf{X} + \mathbf{N}_s,
\end{equation}
where $\mathbf{N}_s$ denotes the receiver noise. To quantify sensing performance, we adopt the mutual information (MI) metric, which measures the information content regarding the target impulse response contained in the received signal \cite{10411942}:
\begin{equation}
\rev{R_s = \log_2 \det \left( \mathbf{I}_{N_r} + \mathbf{G}_{{t}} \mathbf{R}_x \mathbf{G}_{{t}}^H \mathbf{R}_{{i}}^{-1} \right),}
\end{equation}
where $\mathbf{R}_x = \mathbb{E}[\mathbf{x}\mathbf{x}^H]$ is the transmit covariance matrix (approximated as $\mathbf{F}\mathbf{F}^H$ for deterministic beamforming). \rev{The complete interference-plus-noise covariance is}
\begin{equation}
\rev{\mathbf R_i=\mathbf G_{\rm clu}\mathbf R_x\mathbf G_{\rm clu}^H+\mathbf G_{\rm si}\mathbf R_x\mathbf G_{\rm si}^H+\sigma_s^2\mathbf I_{N_r}.}
\end{equation}
\rev{Since $\mathbf G_{\rm clu},\mathbf G_{\rm si}\in\mathbb C^{N_r\times N_t}$ and $\mathbf R_x\in\mathbb C^{N_t\times N_t}$, $\mathbf R_i\in\mathbb C^{N_r\times N_r}$. Thus, the target echo is the useful sensing component, whereas clutter, residual SI, and receiver noise are treated as interference. The noise covariance is already included in $\mathbf R_i$ and is therefore not counted separately in the MI expression.}

% Modeling of time variation
\subsection{Mobility and Time-Varying Channel Model}
To rigorously capture the dynamic nature of the air-ground ISAC network, we jointly model the temporal evolution of the physical kinematics of communication users and sensing targets, and the corresponding time-varying wireless channels.

Let the set of mobile entities include $K$ users and $Q$ sensing targets. We define the kinematic state vector of the $i$-th entity (where $i \in \{1,\dots,K+Q\}$) at time step $t$ as $\mathbf{s}_{i,t} = [\mathbf{p}_{i,t}^T, \mathbf{v}_{i,t}^T]^T$, comprising its 3D Cartesian position $\mathbf{p}_{i,t} \in \mathbb{R}^3$ and velocity $\mathbf{v}_{i,t} \in \mathbb{R}^3$.
The state evolution follows a discrete-time kinematic model \cite{liu2020radar }
\begin{equation}
\mathbf{s}_{i, t+1} = \mathbf{A}_{\mathrm{state}} \mathbf{s}_{i, t} + \mathbf{w}_{i,t},
\end{equation}
where $\mathbf{A}_{\mathrm{state}}$ is the state transition matrix corresponding to the physical motion (e.g., a constant velocity model), and $\mathbf{w}_{i,t}$ represents the process noise.
Consistent with the random-walk mobility model implemented in our simulation, the position evolves deterministically based on velocity ($\mathbf{p}_{t+1} = \mathbf{p}_t + \mathbf{v}_t \Delta t$), while the velocity vector undergoes stochastic perturbations to model unpredictable maneuvers:
\begin{equation}
\mathbf{v}_{i, t+1} = \mathbf{v}_{i, t} + \boldsymbol{\eta}_{v, t},
\end{equation}
where $\boldsymbol{\eta}_{v, t}$ is a noise term bounded by the maximum acceleration capabilities of the user or UAV.

The evolution of the wireless channel exhibits double dynamism driven by both the macroscopic mobility defined above and the microscopic multipath fading. We model the channel components separately for the Line-of-Sight (LoS) and Non-LoS (NLoS) paths.

The LoS path parameters are geometrically determined by the evolving entity positions. At each time step $t$, the LoS AoD $\theta_{i,t}^{{LoS}}$ and complex path gain $g_{i,t}^{{LoS}}$ are updated explicitly \cite{liu2020radar}
\begin{equation}
g_{i,t}^{{LoS}} = \frac{\lambda}{4\pi \|\mathbf{p}_{i,t} - \mathbf{p}_{{BS}}\|} e^{-j \frac{2\pi}{\lambda} \|\mathbf{p}_{i,t} - \mathbf{p}_{{BS}}\|},
\end{equation}
where $\mathbf{p}_{{BS}}$ is the fixed position of the Base Station.

The $L_i - 1$ NLoS scattering paths evolve according to a first-order autoregressive (AR) process, capturing the temporal correlation of the multipath environment. For the $l$-th NLoS path, the complex gain $g_{l,t}$ and angle $\theta_{l,t}$ evolve as \cite{zhang2024predictive}
% (a)(b)
\begin{subequations}
\begin{align}
g_{l,t} &= \rho_g g_{l,t-1} + \sqrt{1-|\rho_g|^2} n_{g,t}, \\
\theta_{l,t} &= \rho_\theta \theta_{l,t-1} + \sqrt{1-\rho_\theta^2} n_{\theta,t},
\end{align}
\end{subequations}
where $n_{g,t} \sim \mathcal{CN}(0, \sigma_g^2)$ and $n_{\theta,t} \sim \mathcal{N}(0, \sigma_\theta^2)$ represent the complex and real process noise, respectively. The correlation coefficient $\rho_g$ is determined by the Jakes' model as $\rho_g = J_0(2\pi f_D \Delta t)$ \cite{new2025channel}, where $J_0(\cdot)$ is the zeroth-order Bessel function of the first kind, $f_D$ is the maximum Doppler frequency induced by the user's velocity, and $\Delta t$ is the time step duration. This formulation ensures that the simulated channel accurately reflects the temporal coherence of the physical environment.
% 閬垮厤proposed
\section{Grey-box BO Method for FAS-ISAC Precoding}
% 鍐楅暱锛岄伩鍏嶇涓€浜虹О
% modeling - solution - extension
This section details the proposed G-MOBO method for FAS-ISAC systems, beginning by formulating the joint design task as a multi-objective optimization problem. To efficiently solve this, we develop a grey-box surrogate model to learn the intermediate physical constituents, which is then coupled with the EHI acquisition function to intelligently navigate the Pareto frontier. To mitigate severe over-exploration within the high-dimensional discrete search space, an adaptive TR mechanism is introduced to dynamically restrict candidate generation into localized hyper-regions. Finally, this static optimization architecture is extended to strictly time-varying environments by incorporating a spatio-temporal tracking strategy.

\subsection{Problem Formulation}
% 1. 涓轰粈涔堥渶瑕乵ulti-objective? 甯︽潵浜嗕粈涔堝洶闅撅紵
% 2. 涓轰粈涔堟槸grey-box
We consider the joint design of a FAS-enabled ISAC system within a generic coherent time slot. The system performance is governed by the configuration of the fluid antenna ports, the macroscopic array orientation, and the digital beamforming. To maintain mathematical rigor in the presence of discrete hardware constraints, we define the global decision variable $\mathbf{z}$ as a combinatorial tuple:
\begin{equation}
\mathbf{z} \triangleq \{\mathcal{M}_t, \mathcal{M}_r, \boldsymbol{\Psi}, \mathbf{F}\},
\end{equation}
where $\mathcal{M}_t, \mathcal{M}_r \subset \mathcal{C}$ are discrete subsets of transmit and receive port indices selected from the candidate grid $\mathcal{C}$. The mechanical orientation $\boldsymbol{\Psi} \triangleq (\theta, \varphi)$ and the beamforming matrix $\mathbf{F}$ are assumed to be selected from finite-resolution codebooks to satisfy practical hardware limitations.

{Traditionally, ISAC optimization relies on a scalarized single-objective framework, where communication and sensing metrics are combined using a predefined weighted sum, $f(\mathbf{z}) = w R_{\mathrm{c}}(\mathbf{z}) + (1-w) R_{\mathrm{s}}(\mathbf{z})$. However, relying on such a scalarized formulation presents significant fundamental limitations. First, Linear scalarization methods intrinsically fail to identify Pareto-optimal configurations that reside within non-convex regions of the objective space. Furthermore, optimizing a fixed scalarization restricts the exploration of different objective trade-offs, which typically leads to poor coverage of the Pareto frontier. To address these limitations, the ISAC design is formulated as a multi-objective optimization problem over the combinatorial discrete space $\mathcal{Z}$, aiming to simultaneously optimize the communication sum-rate $R_{\mathrm{c}}(\mathbf{z})$ and the sensing mutual information $R_{\mathrm{s}}(\mathbf{z})$:
\begin{equation}
\label{eq:static_problem}
\max _{\mathbf{z} \in \mathcal{Z}} \quad \mathbf{f}(\mathbf{z}) \triangleq \big[ R_{\mathrm{c}}(\mathbf{z}),\; R_{\mathrm{s}}(\mathbf{z}) \big]^T
\end{equation}
\vspace{-2em} 
\addtocounter{equation}{-1} % 
\begin{subequations}
\begin{align}
\text{s.t.} \quad
& \operatorname{Tr}(\mathbf{F}\mathbf{F}^H) \leq P_t, \label{eq:st_power}\\
& |\mathcal{M}_t| = N_t, \quad |\mathcal{M}_r| = N_r, \label{eq:st_cardinality}\\
& \mathcal{M}_t \cap \mathcal{M}_r = \varnothing, \label{eq:st_overlap}\\
& \boldsymbol{\Psi} \in \mathcal{Q}_{\Psi}, \quad \mathbf{F} \in \mathcal{Q}_{F}, \label{eq:st_discrete}
\end{align}
\end{subequations}
where \eqref{eq:st_cardinality} and \eqref{eq:st_overlap} enforce RF chain constraints and physical port isolation, respectively, and \eqref{eq:st_discrete} represents the quantization sets for orientation and precoding.} In this multi-objective context, we seek to identify the set of Pareto-optimal configurations. A configuration $\mathbf{z}^\star$ is Pareto-optimal if there exists no other feasible $\mathbf{z} \in \mathcal{Z}$ such that $\mathbf{f}(\mathbf{z}) \succeq \mathbf{f}(\mathbf{z}^\star)$ (where $\succeq$ denotes component-wise inequality with at least one strict inequality). The primary goal is to approximate the Pareto frontier $\mathcal{P} = \{\mathbf{f}(\mathbf{z}) \mid \mathbf{z} \text{ is Pareto-optimal}\}$.

{\color{black}
Solving \eqref{eq:static_problem} through a purely black-box method is inefficient because the mapping from a configuration to the final ISAC metrics contains known analytical structure. Let $\mathbf H_c(\mathbf z)=[\mathbf h_1(\mathbf z),\ldots,\mathbf h_K(\mathbf z)]$ collect the communication channels associated with configuration $\mathbf z$. For a complex matrix $\mathbf A$, define $\operatorname{rvec}(\mathbf A)=[\Re\{\operatorname{vec}(\mathbf A)\}^T,\Im\{\operatorname{vec}(\mathbf A)\}^T]^T$. The intermediate vector learned by the grey-box surrogate is the concatenation of three physically interpretable blocks:
\begin{equation}
\begin{aligned}
\mathbf h(\mathbf z)&=\left[\mathbf h_{\rm c}^T(\mathbf z),\mathbf h_{\rm t}^T(\mathbf z),\mathbf h_{\rm I}^T(\mathbf z)\right]^T,\\
\mathbf h_{\rm c}(\mathbf z)&=\operatorname{rvec}(\mathbf H_c(\mathbf z)),\\
\mathbf h_{\rm t}(\mathbf z)&=\operatorname{rvec}(\mathbf G_t(\mathbf z)),\\
\mathbf h_{\rm I}(\mathbf z)&=\left[\operatorname{rvec}(\mathbf G_{\rm clu}(\mathbf z))^T,
\operatorname{rvec}(\mathbf G_{\rm si}(\mathbf z))^T\right]^T.
\end{aligned}
\label{eq:greybox_constituents}
\end{equation}
Thus, every learned intermediate quantity is directly configuration dependent: $\mathbf h_{\rm c}$ describes the communication channel, $\mathbf h_{\rm t}$ describes the target response, and $\mathbf h_{\rm I}$ combines the clutter and residual-SI responses. Since $\mathbf F$ is contained in $\mathbf z$, the transmit covariance $\mathbf R_x(\mathbf z)=\mathbf F\mathbf F^H$ is also known for each evaluated configuration. After reversing $\operatorname{rvec}(\cdot)$, the deterministic mapping $\mathbf g(\cdot)$ computes the final objective vector as
\begin{equation}
\begin{aligned}
\mathbf f(\mathbf z)&=\mathbf g(\mathbf h(\mathbf z))=\left[R_c(\mathbf z),R_s(\mathbf z)\right]^T,\\
R_c(\mathbf z)&=\sum_{k=1}^K\log_2\!\left(1+\frac{|\mathbf h_k^T(\mathbf z)\mathbf f_k|^2}{\sum_{j\ne k}|\mathbf h_k^T(\mathbf z)\mathbf f_j|^2+\sigma_c^2}\right),\\
R_s(\mathbf z)&=\log_2\det\!\left(\mathbf I_{N_r}+\mathbf G_t(\mathbf z)\mathbf R_x(\mathbf z)\mathbf G_t^H(\mathbf z)\mathbf R_i^{-1}(\mathbf z)\right),\\
\mathbf R_i(\mathbf z)&=\mathbf G_{\rm clu}(\mathbf z)\mathbf R_x(\mathbf z)\mathbf G_{\rm clu}^H(\mathbf z)
+\mathbf G_{\rm si}(\mathbf z)\mathbf R_x(\mathbf z)\mathbf G_{\rm si}^H(\mathbf z)\\
&\quad+\sigma_s^2\mathbf I_{N_r}.
\end{aligned}
\label{eq:greybox_mapping}
\end{equation}
The surrogate therefore learns the configuration-dependent physical intermediate quantities in \eqref{eq:greybox_constituents}, whereas the multi-objective performance is calculated through the known analytical mapping in \eqref{eq:greybox_mapping}.

The known or calibrated components are the candidate-port geometry, selected port coordinates, steering-response functional form, precoder-induced $\mathbf R_x(\mathbf z)$, power and noise parameters, and nominal distance-dependent SI structure. At each evaluated configuration, pilot-assisted communication measurements and sensing echo samples, or equivalent calibrated feedback, provide observations of $\mathbf H_c$, $\mathbf G_t$, $\mathbf G_{\rm clu}$, and the effective residual $\mathbf G_{\rm si}$ used to update the surrogate. Thus, these configuration-dependent constituents are not assumed known a priori. Moderate mismatch in the nominal propagation or SI components can be absorbed by newly collected constituent observations, whereas severe error in the deterministic mapping $\mathbf g(\cdot)$ may reduce the advantage of the grey-box decomposition; no universal model-mismatch robustness guarantee is claimed.
}

Furthermore, in dynamic air-ground networks, the environment is time-variant. The objectives evolve as $\mathbf{f}_t(\mathbf{z}) = [R_{\mathrm{c},t}(\mathbf{z}),\; R_{\mathrm{s},t}(\mathbf{z})]^T$, causing the Pareto set $\mathcal{P}_t$ to drift. Thus, the problem extends to a sequential decision-making task: continuously tracking high-quality configurations over $T$ slots while navigating a high-dimensional, combinatorial search space without gradient information.

\subsection{Surrogate Modeling and Acquisition}
Given the high-dimensional and discrete nature of the FAS configuration space $\mathcal{Z}$, standard BO methods often suffer from over-exploration, wasting function evaluations in regions with low probability of improvement. To overcome this, we propose a grey-box multi-objective BO approach for the static optimization phase (i.e., solving Problem \eqref{eq:static_problem} at $t=0$). The core principle is to restrict the search for the next candidate $\mathbf{z}_{m+1}$ within a local TR centered around the current best solution $\mathbf{z}^*$, rather than searching the entire global space.

\subsubsection{GP Surrogate}
\rev{A GP supplies an analytically tractable probabilistic realization of the grey-box surrogate, as illustrated in Fig.~\ref{fig:surrogate_alignment}. It assumes that every finite collection of latent constituent values follows a joint multivariate Gaussian distribution. Because the FAS configuration space $\mathcal{Z}$ contains discrete port indices and quantized beamforming parameters, the GP realization embeds the hybrid configuration $\mathbf{z}$ into a continuous latent vector $\mathbf{x}=\phi(\mathbf{z})\in\mathbb{R}^d$, using, for example, one-hot encodings and normalized continuous features. This realization provides the calibrated reference moments used in the reproducing kernel Hilbert space (RKHS) analysis below; it is not itself the defining contribution of G-MOBO.}
The correlation between any two configurations $\mathbf{z}_i$ and $\mathbf{z}_j$ is then modeled by a stationary kernel $\mathcal{K}(\mathbf{z}_i, \mathbf{z}_j)$, such as the Matern kernel or the Radial Basis Function (RBF)\cite{lu2023surrogate}
\begin{equation}
\mathcal{K}(\mathbf{z}_i, \mathbf{z}_j) = \exp\left(-\frac{1}{2} \|\phi(\mathbf{z}_i) - \phi(\mathbf{z}_j)\|^2\right).
\end{equation}

Given the observed data $\mathcal{D}_m=\{(\mathbf{z}_i,\hat{\mathbf{h}}_i)\}_{i=1}^m$, we place a GP prior on each constituent and model $\mathbf{h}(\mathbf{z})$ as a (possibly multi-output) GP with posterior mean $\boldsymbol{\mu}_h(\mathbf{z})$ and covariance $\mathbf{K}_h(\mathbf{z},\mathbf{z}')$, where $\mathbf{K}$ is the $m \times m$ kernel matrix with entries $[\mathbf{K}]_{ij} = \mathcal{K}(\mathbf{z}_i, \mathbf{z}_j)$. To predict the value at a new candidate point $\mathbf{z}_{m+1}$, we exploit the joint Gaussian property
\begin{equation}
\begin{bmatrix}
\mathbf{h}(\mathbf{z}_1) \\
\vdots \\
\mathbf{h}(\mathbf{z}_m) \\
\mathbf{h}(\mathbf{z}_{m+1})
\end{bmatrix}
\sim \mathcal{N}\!\left(\mathbf{0},
\begin{bmatrix}
\mathbf{K} & \mathbf{k}\\
\mathbf{k}^T & \mathcal{K}(\mathbf{z}_{m+1},\mathbf{z}_{m+1})
\end{bmatrix}\right),
\end{equation}
where $\mathbf{k} = [\mathcal{K}(\mathbf{z}_{m+1}, \mathbf{z}_1), \dots, \mathcal{K}(\mathbf{z}_{m+1}, \mathbf{z}_m)]^T$. The conditional posterior $p(\mathbf{h}(\mathbf{z}_{m+1})\mid \mathcal{D}_m,\mathbf{z}_{m+1})$ is Gaussian with mean $\boldsymbol{\mu}_h(\mathbf{z}_{m+1})$ and covariance $\boldsymbol{\Sigma}_h(\mathbf{z}_{m+1})$
\cite{kandasamy2018parallelised}
\begin{subequations}   
\begin{align}
\boldsymbol{\mu}_h(\mathbf{z}_{m+1}) &= \mathbf{k}^T \mathbf{K}^{-1} \mathbf{y}_{1:m} \\
\boldsymbol{\Sigma}_h(\mathbf{z}_{m+1}) &= \mathcal{K}(\mathbf{z}_{m+1}, \mathbf{z}_{m+1}) - \mathbf{k}^T \mathbf{K}^{-1} \mathbf{k}.
\end{align}
\end{subequations}
\rev{GPs provide principled uncertainty quantification, but the implemented realization has two practical limitations in the FAS setting: its exact posterior update scales as $\mathcal O(m^3)$, and a stationary kernel applied after continuous relaxation may impose correlations that do not match discrete port indices and quantized phases. A common wrapper maps discrete configurations into a continuous space before evaluating the covariance \cite{garrido2020dealing}, yielding $\mathcal K'(\mathbf z_i,\mathbf z_j)=\mathcal K(T(\mathbf z_i),T(\mathbf z_j))$. This mapping permits GP inference but does not guarantee that the selected kernel accurately represents the resulting combinatorial response.}

\begin{figure}[t]
    \centering

\includegraphics[width=\linewidth]{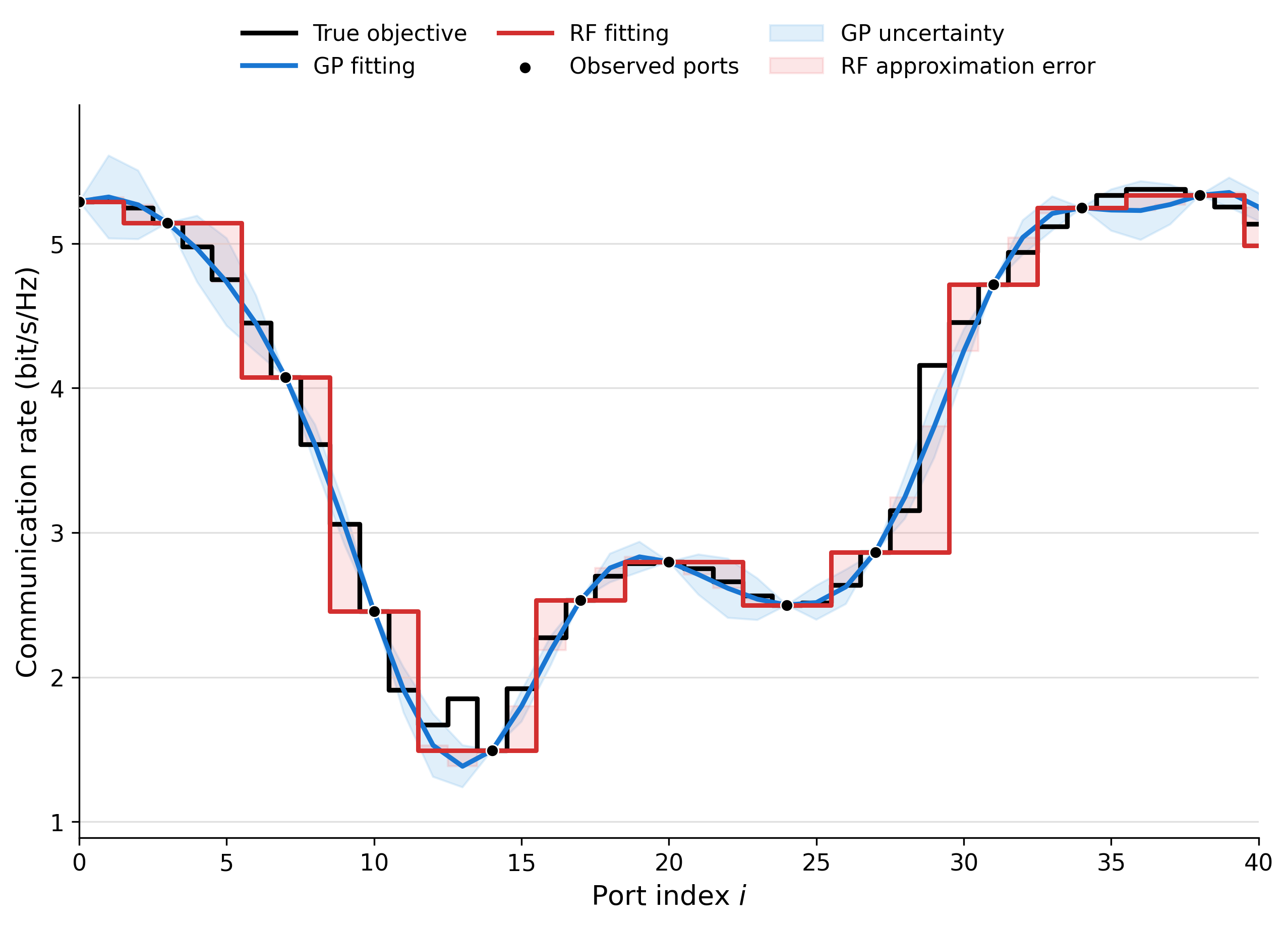}

    \caption{\rev{GP and RF surrogate realizations over discrete port indices, with their respective uncertainty and approximation-error bounds.}}
    \label{fig:surrogate_alignment}
\end{figure}
\subsubsection{ RF Surrogate}
\rev{RF is used as a practical surrogate realization for the discrete variables in $\mathcal Z$ \cite{pmlr-v119-ru20a}. Its tree partitions operate directly on port indices and quantized codebook variables, while the surrounding grey-box decomposition, EHI acquisition, and TR search remain unchanged. Applied componentwise to $\mathbf h(\mathbf z)$, the RF consists of an ensemble of $M$ regression trees. This tree-based structure partitions the combinatorial FAS configuration space into hyper-rectangular leaf regions and naturally accommodates the port-selection constraints in \eqref{eq:st_cardinality} and \eqref{eq:st_overlap}, as well as the quantized codebook constraints in \eqref{eq:st_discrete}.}

\rev{For a given latent constituent, let $\mathcal T_b(\mathbf z)$ denote the prediction of the $b$-th tree at configuration $\mathbf z$. Its piecewise-constant prediction can be written as}
\begin{equation}
\rev{\mathcal T_b(\mathbf z)
=\sum_{j=1}^{J_b}c_{j,b}\,\mathbb I(\mathbf z\in\mathcal R_{j,b}),}
\label{eq:rf_tree_prediction}
\end{equation}
\rev{where $\mathcal R_{j,b}$ is the region associated with the $j$-th leaf of tree $b$, $J_b$ is the number of its leaves, and $c_{j,b}$ is the empirical mean of the observed constituent values assigned to that leaf. The forest aggregates these individual tree predictions to approximate the constituent posterior. Its predictive mean and variance are given by \cite{hutter2011sequential}}
\begin{subequations}
\label{eq:rf_moments}
\begin{align}
    \widehat\mu^{\rm RF}(\mathbf{z}) &= \frac{1}{M}\sum_{m=1}^{M} \mathcal{T}_m(\mathbf{z}), \label{eq:rf_mean} \\
    \big(\widehat\sigma^{\rm RF}(\mathbf{z})\big)^2 &= \frac{1}{M-1}\sum_{m=1}^{M} \left(\mathcal{T}_m(\mathbf{z}) - \widehat\mu^{\rm RF}(\mathbf{z})\right)^2. \label{eq:rf_var}
\end{align}
\end{subequations}
\rev{The ensemble variance measures the disagreement among the regression trees and serves as a practical proxy for epistemic uncertainty. These RF moments are propagated through the known mapping $g(\cdot)$ to evaluate EHI. In the high-dimensional discrete configuration space, RF therefore provides a computationally practical approximation of the GP predictive moments while operating directly on port-index variables, as illustrated in Fig.~\ref{fig:surrogate_alignment}.}

\subsubsection{EHI Acquisition Function}
After modeling the latent constituent vector $\mathbf{h}(\mathbf{z}) \in \mathbb{R}^{K_h}$ via grey-box surrogate models, the EHI is adopted as the acquisition function for the following reasons. First, EHI directly drives the search towards the true Pareto frontier without relying on arbitrary scalarization weights. Furthermore, it is practically highly efficient because our ISAC problem inherently comprises only two objectives, which strictly bounds the computational overhead. Formally, EHI is defined as the expected increase in the hypervolume within the objective space $\mathbf{f}(\mathbf{z}) = [R_{\mathrm c}(\mathbf{z}),\, R_{\mathrm s}(\mathbf{z})]^{\!\top}$, measured with respect to a reference point $\mathbf{r}$. For a given configuration $\mathbf{z}$ and the obtained objective vector $\mathbf{f}(\mathbf{z})$, the hypervolume improvement (HVI) is calculated as
\begin{equation}
\mathrm{HVI}(\mathbf f;\mathcal P_n,\mathbf r)
\triangleq
\mathrm{HV}(\mathcal P_n\cup\{\mathbf f\},\mathbf r)
-
\mathrm{HV}(\mathcal P_n,\mathbf r),
\end{equation}
where the hypervolume (HV), $\mathrm{HV}(\cdot,\mathbf r)$, denotes the Lebesgue measure of the region dominated by the Pareto set and bounded by the reference point.

Under the grey-box ISAC formulation
$\mathbf f(\mathbf z) = \mathbf g(\mathbf h(\mathbf z))$,
the predictive distribution of $\mathbf f(\mathbf z)$ is induced by the posterior
$p(\mathbf h \mid \mathbf z, \mathcal D_n)$
of the latent constituent vector. Consequently, the EHI acquisition function at iteration $n$ is given by
\begin{equation}
\label{eq:ehvi_cf}
\begin{aligned}
\alpha_n^{\mathrm{EHI}}(\mathbf z)
&\triangleq
\mathbb E_{\mathbf h \sim p(\mathbf h \mid \mathbf z, \mathcal D_n)}
\!\left[
\mathrm{HVI}\!\left(\mathbf g(\mathbf h);\, \mathcal P_n, \mathbf r\right)
\right]. \\
% &=
% \int_{\mathbb R^{K_h}}
% \mathrm{HVI}\!\left(\mathbf g(\mathbf h);\, \mathcal P_n, \mathbf r\right)\,
% p(\mathbf h \mid \mathbf z, \mathcal D_n)\, d\mathbf h 
\end{aligned}
\end{equation}
% (26) 鍗曡锛?(28)->(29) 杩涜璇︾粏鎺ㄥ
When the grey-box surrogate provides a Gaussian predictive posterior over the latent constituent vector,
\begin{equation}
\mathbf h(\mathbf z)\mid \mathcal D_n \sim \mathcal N\!\big(\boldsymbol\mu_h(\mathbf z),\, \boldsymbol\Sigma_h(\mathbf z)\big),
\end{equation}
the reparameterization trick is applied to express samples of $\mathbf h$ in a differentiable form. Let $\boldsymbol\Sigma_h(\mathbf z)=\mathbf L_h(\mathbf z)\mathbf L_h^\top(\mathbf z)$ be the Cholesky factorization. Then the reparameterization of ${\mathbf h}$ can be defined as $\tilde{\mathbf h} = \boldsymbol\mu_h(\mathbf z) + \mathbf L_h(\mathbf z)\boldsymbol\epsilon$, where $\boldsymbol\epsilon\sim\mathcal N(\mathbf 0,\mathbf I)$. Substituting this into \eqref{eq:ehvi_cf}, the EHI can be equivalently expressed as an expectation over the standard normal distribution:
\begin{equation}
\label{eq:ehvi_cf_reparam}
\begin{aligned}
{\alpha}_n^{\mathrm{EHI}}(\mathbf z)
&=
\mathbb E_{\boldsymbol\epsilon\sim\mathcal N(\mathbf 0,\mathbf I)}\!\left[
\mathrm{HVI}\!\left(\mathbf g(\tilde{\mathbf h});\, \mathcal P_n, \mathbf r\right)
\right]\\
&=
\int_{\mathbb{R}^{K_h}}
\mathrm{HVI}\!\left(\mathbf g\big(\boldsymbol\mu_h(\mathbf z) + \mathbf L_h(\mathbf z)\boldsymbol\epsilon\big);\, \mathcal P_n, \mathbf r\right)
p(\boldsymbol\epsilon)\, d\boldsymbol\epsilon,
\end{aligned}
\end{equation}
where $p(\boldsymbol\epsilon)$ is the standard multivariate normal probability density function.

In the established grey-box ISAC setting, the mapping
$\mathbf f(\mathbf{z}) = \mathbf g(\mathbf h(\mathbf{z}))$
is highly non-linear due to the logarithmic SINR and log-determinant mutual information operators. As the induced density of the objective values lacks a closed-form expression, the integral in \eqref{eq:ehvi_cf_reparam} is numerically approximated via sample average approximation \cite{Astudillo2021GreyBoxBO}. By drawing $N_{\mathrm{mc}}$ samples $\{\boldsymbol\epsilon^{(j)}\}_{j=1}^{N_{\mathrm{mc}}}$ from $\mathcal N(\mathbf 0,\mathbf I)$, the EHI is computed as \cite{Yang2019EHVI}
\begin{equation}
\hat{\alpha}_n^{\mathrm{EHI}}(\mathbf z)
=
\frac{1}{N_{\mathrm{mc}}}
\sum_{j=1}^{N_{\mathrm{mc}}}
\mathrm{HVI}\!\left(
\mathbf g\!\left(\tilde{\mathbf h}^{(j)}\right);
\mathcal P_n, \mathbf r
\right),
\end{equation}
where $\tilde{\mathbf h}^{(j)} = \boldsymbol\mu_h(\mathbf z) + \mathbf L_h(\mathbf z)\boldsymbol\epsilon^{(j)}$. This formulation enables the uncertainty of the underlying channel and SI constituents to be propagated through the known ISAC performance models, yielding a physically consistent acquisition function that drives sample-efficient exploration of the sensing-communication Pareto frontier.

\begin{algorithm}[t]
\caption{G-MOBO for FAS-ISAC in Static Optimization}
\label{alg:morbo}
\begin{algorithmic}[1]
\STATE \textbf{Input:} Evaluation budget $B$, initial sample size $N_0$, reference point $\mathbf{r}$, TR parameters $(L_{\min}, L_{\max}, \gamma_{\mathrm{i}}, \gamma_{\mathrm{d}})$.
\STATE Generate $N_0$ initial configurations $\{\mathbf{z}_i\}_{i=1}^{N_0}$ satisfying constraints \eqref{eq:st_cardinality}, \eqref{eq:st_overlap}, and \eqref{eq:st_discrete}.
\STATE Evaluate objective $\mathbf{f}(\mathbf{z}_i)$ and extract latent constituents $\hat{\mathbf{h}}_i$.
\STATE Initialize dataset $\mathcal{D}_0 \gets \{(\mathbf{z}_i, \hat{\mathbf{h}}_i, \mathbf{f}_i)\}_{i=1}^{N_0}$ and Pareto set $\mathcal{P}_0$.
\STATE Initialize TR $\mathcal{T}_0$ with side-length $L_0 = L_{\max}$ centered at a configuration in $\mathcal{P}_0$.
\FOR{$n = 0$ to $B-1$}
    \STATE \rev{Fit a grey-box surrogate $\mathcal S\in\{\mathcal S_{\rm GP},\mathcal S_{\rm RF}\}$ for $\mathbf{h}(\mathbf{z})$ using $\mathcal{D}_n$; use \eqref{eq:rf_moments} for the practical RF realization.}
    \STATE Acquire candidate $\mathbf{z}_{\mathrm{n+1}}$ by maximizing the constrained acquisition function \eqref{eq:constrained_EHVI}:
    \STATE \quad $\mathbf{z}_{\mathrm{n+1}} = \arg\max_{\mathbf{z} \in \mathcal{T}_n} \alpha^{\mathrm{EHI}}_{C,n}(\mathbf{z})$
    \STATE Evaluate objective $\mathbf{f}_{\mathrm{n+1}}$ and extract constituent $\hat{\mathbf{h}}_{\mathrm{n+1}}$.
    \STATE Update dataset $\mathcal{D}_{n+1} \gets \mathcal{D}_n \cup \{(\mathbf{z}_{\mathrm{n+1}}, \hat{\mathbf{h}}_{\mathrm{n+1}}, \mathbf{f}_{\mathrm{n+1}})\}$ and Pareto set $\mathcal{P}_{n+1}$.
    \STATE Update TR side-length $L_{n+1}$ via \eqref{eq:L_update}.
    \IF{$L_{n+1} < L_{\min}$}
        \STATE Restart $\mathcal{T}_{n+1}$ centered at a configuration in $\mathcal{P}_{n+1}$ and reset $L_{n+1} \gets L_{\max}$.
    \ENDIF
\ENDFOR
\STATE \textbf{Output:} Final approximated Pareto set $\mathcal{P}^*$.
\end{algorithmic}
\end{algorithm}

\subsection{Constrained Local Search via TR Adaptation}
Given the high-dimensional and combinatorial nature of the FAS configuration space $\mathcal{Z}$ comprising the discrete tuple $\{\mathcal{M}_t, \mathcal{M}_r, \boldsymbol{\Psi}, \mathbf{F}\}$, standard BO methods inevitably suffer from severe over-exploration, wasting valuable function evaluations in regions with a low probability of improvement. To overcome this fundamental bottleneck, a localized optimization architecture is developed, relying on two sequentially integrated mechanisms: adaptive local TRs and strict constraint handling. Rather than searching the entire global space, the optimization restricts candidate generation within localized hyper-regions. Unlike single-objective optimization that centers a search around a unique global optimum, these TRs are dynamically centered around specific Pareto-optimal configurations exhibiting high hypervolume contributions. This localized tracking is then tightly coupled with a hard-rejection mechanism to ensure all generated candidates strictly satisfy the physical system constraints.

\subsubsection{Pareto Frontier Tracking via Adaptive TR Resizing}
A critical mechanism of the localized search is the adaptive adjustment of the TR side-length $L_n$, which dynamically balances local exploitation and global exploration within the discrete search space. The region size is updated based on the hypervolume improvement achieved during the current iteration \cite{Daulton2022MORBO}. In traditional single-objective optimization, a search is considered successful merely if the scalar objective value increases. In stark contrast, within the competitive dual-objective ISAC environment, a successful search is strictly defined as discovering a new configuration $\mathbf{z}_{\mathrm{new}}$ that strictly dominates existing solutions or uncovers a completely new optimal trade-off. Mathematically, this multi-objective progression is quantified by a strictly positive hypervolume improvement, defined as $\mathrm{HVI}(\mathbf{f}(\mathbf{z}_{\mathrm{new}}); \mathcal{P}_n, \mathbf{r}) > 0$. The update rule for the TR side-length $L_{n+1}$ is formulated as:
\begin{equation}
\label{eq:L_update}
L_{n+1} =
\begin{cases}
\min(L_{n} \cdot \gamma_{\mathrm{i}}, L_{\max}), & \text{if } C_s \ge \tau_s \\
\max(L_{n} \cdot \gamma_{\mathrm{d}}, L_{\min}), & \text{if } C_f \ge \tau_f \\
L_{n}, & \text{otherwise}
\end{cases},
\end{equation}
where $\gamma_{\mathrm{i}} > 1$ and $\gamma_{\mathrm{d}} < 1$ are the expansion and contraction factors, respectively. The variables $C_s$ and $C_f$ denote the counters for consecutive successes and failures. The TR expands to allow bolder steps toward unexplored areas only when the optimizer consistently advances the Pareto frontier for $\tau_s$ consecutive iterations. Conversely, the region shrinks to enforce fine-grained local refinement only after the algorithm fails to improve the hypervolume for $\tau_f$ consecutive iterations. Otherwise, the region size remains unchanged to permit sufficient exploration of the current locality. This mechanism prevents premature resizing due to stochastic fluctuations and ensures that the surrogate models maintain high local accuracy within the vast combinatorial space.

\subsubsection{Enforcing Hardware Constraints in Candidate Acquisition}
At each optimization iteration, the search for candidate configurations is strictly bounded within the current TR. To guarantee that the generated samples physically adhere to the hardware limitations defined in \eqref{eq:st_cardinality} and \eqref{eq:st_overlap}, a hard-rejection acquisition strategy is employed. The multi-objective EHI acquisition function is augmented by incorporating a feasibility indicator $\mathbb{I}(\mathbf{z} \in \mathcal{F})$. This indicator explicitly verifies deterministic geometric and hardware conditions, such as ensuring the transmit and receive port subsets satisfy the exact RF chain cardinality constraints ($|\mathcal{M}_t| = N_t$ and $|\mathcal{M}_r| = N_r$) and maintain strict physical isolation ($\mathcal{M}_t \cap \mathcal{M}_r = \varnothing$). The constrained acquisition function $\alpha_C(\mathbf{z})$ is defined as:
\begin{equation}
\alpha^{\mathrm{EHI}}_{C,n}(\mathbf{z}) = \hat{\alpha}_n^{\mathrm{EHI}}(\mathbf{z}) \cdot \mathbb{I}(\mathbf{z} \in \mathcal{F}), \label{eq:constrained_EHVI}
\end{equation}
where $\mathbb{I}(\mathbf{z} \in \mathcal{F}) = 1$ if the configuration $\mathbf{z}$ satisfies all structural constraints outlined in \eqref{eq:static_problem}, and 0 otherwise. This formulation ensures that the optimizer assigns zero utility to physically unrealizable designs, automatically filtering out overlapping ports or invalid antenna orientations. By multiplying EHI with this exact hard-rejection mask, the optimization framework exclusively concentrates its limited evaluation budget on the feasible portions of the discrete search space, thereby facilitating a highly sample-efficient exploration of the sensing and communication trade-off. The detailed procedure is summarized in Algorithm~\ref{alg:morbo}.
\begin{algorithm}[t]
\caption{Extended G-MOBO for Dynamic FAS-ISAC}
\label{alg:tv_morbo}
\begin{algorithmic}[1]
\STATE \textbf{Input:} Time horizon $T$, per-slot budget $M$, window capacity $W$, initial dataset $\mathcal{D}_0$, initial Pareto set $\mathcal{P}_0$.
\FOR{$t = 1$ to $T$}
    \STATE Update localized dataset $\mathcal{D}_t$ via \eqref{eq:dataset_update} by appending fresh observations and discarding obsolete samples.
    \STATE Fit spatio-temporal grey-box surrogates for $\mathbf{h}(\mathbf{z})$ exclusively using $\mathcal{D}_t$.
    \STATE Initialize TR $\mathcal{T}_{t,0}$ centered at a Pareto-optimal configuration $\mathbf{z}_{c, t} \in \mathcal{P}_{t-1}$.
    \STATE Reset TR side-length $L_{t,0} \gets L_{\max}$.
    \STATE Execute the constrained candidate acquisition and TR update procedures (equivalent to Lines 7 to 15 in Algorithm \ref{alg:morbo}) for $M$ iterations.
    \STATE Store the final approximated Pareto set for the current time slot: $\mathcal{P}_t \gets \mathcal{P}_{t,M}$.
\ENDFOR
\STATE \textbf{Output:} Sequence of Pareto frontiers $\{\mathcal{P}_t\}_{t=1}^T$.
\end{algorithmic}
\end{algorithm}

% recall air-ground
% 閽堝浠€涔堝叧閿瘝
% channel 寮哄寲涓€涓?
% physical layer銆?鍐呭鍜岃〃杩?
% 
\subsection{Adaptation to Time-Varying Environment}
As established previously, practical air ground networks are characterized by severe dynamics stemming from the continuous mobility of communication users and sensing targets. Consequently, the end-to-end ISAC objective function $\mathbf{f}_t(\mathbf{z})$ evolves continuously over time due to rapid channel aging and spatial topological changes. The standard G-MOBO  algorithm operates under a quasi-static assumption, aggregating all historical data to refine a purely spatial surrogate model. Applying this static framework directly to highly dynamic scenarios causes the optimizer to misinterpret environmental drift as stochastic noise. As a result, a configuration that is Pareto-optimal at time $t-1$ may degrade significantly at time $t$, inevitably leading to search stagnation. To overcome this fundamental limitation, the optimization architecture is extended to a time-varying design. This extension incorporates explicit temporal awareness into the surrogate modeling to robustly track the drifting optimum. Specifically, the GP surrogate kernel function is extended from a purely spatial domain to a joint spatio-temporal domain. The spatio-temporal kernel is defined as the product of a spatial kernel $\mathcal{K}_s$ and a temporal kernel $\mathcal{K}_t$:
\begin{equation}
\mathcal{K}((\mathbf{z}, t), (\mathbf{z}', t')) = \mathcal{K}_s(\mathbf{z}, \mathbf{z}') \cdot \mathcal{K}_t(t, t'),
\end{equation}
where $\mathcal{K}_s$ typically employs the standard Matern kernel. The temporal kernel $\mathcal{K}_t$ quantifies the degradation of historical information over time, which is mathematically modeled using an exponential decay function:
\begin{equation}
\mathcal{K}_t(t, t') = \exp\left( -\frac{|t - t'|}{\tau} \right),
\end{equation}
\rev{where $\tau$ is the characteristic time-scale controlling the temporal correlation limit. This allows the model to capture global structures through historical trends while remaining sensitive to instantaneous changes in the environment.}

Parallel to the temporal kernel approach, the RF surrogate inherently accommodates time-varying objectives through its recursive partitioning mechanism. By incorporating the time index $t$ into the input feature vector $\mathbf{x} = [\mathbf{z}^\top, t]^\top$, the RF adaptively partitions the search space into spatio-temporal hyper-rectangles. Each node in a decision tree performs an optimal split based on either a spatial configuration coordinate or the timestamp $t$, yielding disjoint regions formulated as:
\begin{equation}
\mathcal{R}_{\mathrm{left}} = \{\mathbf{x} \mid x_i \leq s\}, \quad \mathcal{R}_{\mathrm{right}} = \{\mathbf{x} \mid x_i > s\}.
\end{equation}
In dynamic settings, this tree structure effectively isolates obsolete data into specific branches using the temporal features. The predictive output for the latent constituents at a given spatio-temporal coordinate is then aggregated over the leaf nodes:
\begin{equation}
\hat{\mathbf{h}}(\mathbf{z}, t) = \sum_{l=1}^{L} \bar{\mathbf{y}}_l \cdot \mathbb{I}\big((\mathbf{z}, t) \in \mathcal{L}_l\big),
\end{equation}
where $\mathcal{L}_l$ represents the $l$-th leaf node region, $\bar{\mathbf{y}}_l$ is the sample mean vector within that node, and $\mathbb{I}(\cdot)$ is the indicator function. As new observations arrive, the RF dynamically forms new splits at the leaf nodes to capture the instantaneous landscape of the intermediate physical constituents. This localized refinement enables the surrogate to maintain high predictive accuracy without the over-smoothing bias often associated with stationary kernels.

To continuously track the drifting Pareto frontier across time slots, an explicit inter-slot TR reset is performed. At the beginning of each new time slot $t$, while the TR center $\mathbf{z}_{c, t}$ is anchored to a configuration from the previous Pareto set $\mathcal{P}_{t-1}$, its side-length $L_{t,0}$ is fully reset to the maximum bound $L_{\max}$ to ensure sufficient exploration of the updated environment. This temporal reset is formulated as:
$$
\mathcal{T}_{t,0} = \Big\{ \mathbf{z} \in \mathcal{Z} \;\Big|\; \text{dist}(\mathbf{z}, \mathbf{z}_{c, t}) \leq L_{\max} \Big\}, \quad \text{where } \mathbf{z}_{c, t} \in \mathcal{P}_{t-1}.
$$

Furthermore, to prevent the surrogate models from being misled by severely outdated historical channel states, a sliding window memory mechanism is explicitly implemented. Unlike the monotonically cumulative dataset utilized in static optimization, a localized historical dataset $\mathcal{D}_t$ is maintained with a strictly bounded capacity $W$. At the transition to time slot $t$, the dataset is dynamically updated by appending the fresh observations evaluated during the previous slot ($\mathcal{D}_{\mathrm{new}}$) and discarding the obsolete samples that exceed the temporal window ($\mathcal{D}_{\mathrm{old}}$):
\begin{equation} \label{eq:dataset_update}
\mathcal{D}_t = \big( \mathcal{D}_{t-1} \cup \mathcal{D}_{\mathrm{new}} \big) \setminus \mathcal{D}_{\mathrm{old}}
\end{equation}
By subsequently re-training the grey box surrogate models exclusively on this truncated dataset $\mathcal{D}_t$, the framework effectively bounds the computational complexity while robustly capturing the instantaneous landscape of the current objective $\mathbf{f}_t(\mathbf{z})$. The detailed procedure is summarized in Algorithm~\ref{alg:tv_morbo}.
 
\section{Theoretical Analysis}
% 鍥寸粫BO-framework
% 璁ㄨ浜哫X
This section provides a theoretical analysis on Algorithm 1 for the proposed BO framework, with a primary focus on evaluating the sample efficiency. The discussion characterizes how the sample complexity within the high-dimensional combinatorial search space can be reduced via the localized TR strategy and the grey-box surrogate modeling. \rev{The analysis first establishes a GP/RKHS reference result, for which calibrated posterior confidence bounds are available. It then connects the practical RF realization to this reference through an explicit bounded predictive-moment approximation assumption. The resulting RF-compatible characterization retains the GP/RKHS and TR terms and adds a separate approximation penalty; therefore, it is conditional and does not assert that an arbitrary RF automatically inherits a GP upper confidence bound (GP-UCB) guarantee.}
\subsection{Regret Bound Analysis}\label{sec:regret_bound}

\rev{The fundamental optimization performance of the GP-based instantiation is first analyzed via cumulative regret bounds.} Let $\mathbf{f}(\mathbf{z})=[f_1(\mathbf{z}),\dots,f_M(\mathbf{z})]^\top \in \mathbb{R}^M$ denote the unknown objective vector defined over the domain $\mathcal{Z}$, where the goal is component-wise maximization. A reference point $\mathbf{r}\in\mathbb{R}^M$ is fixed. For any finite set $\mathcal{P}\subset\mathbb{R}^M$, let $HV(\mathcal{P};\mathbf{r})$ denote the dominated hypervolume with respect to $\mathbf{r}$. Given a sequence of queried points $\mathcal{Z}_N = \{\mathbf{z}_n\}_{n=1}^N$, the corresponding evaluated image set is defined as $\mathcal{P}_n=\{\mathbf{f}(\mathbf{z}_1),\dots,\mathbf{f}(\mathbf{z}_n)\}$. Let $HV_n \triangleq HV(\mathcal{P}_n;\mathbf{r})$ be the hypervolume covered at step $n$, which is strictly non-decreasing. Let $\mathcal{P}^\star$ denote the unknown Pareto-optimal set and $HV^\star \triangleq HV(\mathcal{P}^\star;\mathbf{r})$ be the maximum achievable hypervolume. The instantaneous hypervolume regret $\Delta_n^{HV}$ and the cumulative hypervolume regret $R_N^{HV}$ are defined respectively as:
\begin{equation}
\Delta_n^{HV} \triangleq HV^\star - HV_n, \quad R_N^{HV} \triangleq \sum_{n=1}^{N} \Delta_n^{HV}.
\label{eq:hv_cum_regret_def}
\end{equation}

To systematically bound the regret, the regularity of the objective functions and the hypervolume decomposition metric are formalized in two assumptions.

\begin{assumption}\label{ass:function_regularity}
For each objective $m\in\{1,\dots,M\}$, the function $f_m$ resides in the RKHS $\mathcal{H}_k$ associated with kernel $k$, satisfying the strict norm bound $\|f_m\|_{\mathcal{H}_k}\le B_m$. The observations are subject to additive noise $y_{m,n}=f_m(\mathbf{z}_n)+\epsilon_{m,n}$, where the noise terms $\epsilon_{m,n}$ are independent and $\sigma$-sub-Gaussian.
\end{assumption}

\begin{assumption}\label{ass:hypervolume_decomposition}
The multi-objective hypervolume deficiency can be upper-bounded via a finite set of scalarized sub-problems. Specifically, there exist non-negative coefficients $\{\alpha_k\}_{k=1}^K$ and a discretization error $\varepsilon_K \ge 0$ (satisfying $\lim_{K\to\infty} \varepsilon_K = 0$) corresponding to a fixed set of weight vectors $\{\mathbf{w}_k\}_{k=1}^K \subset \mathbb{R}_+^M$, such that the hypervolume deficiency for any evaluated set $\mathcal{Z}_n$ is bounded by $HV^\star - HV_n \le \varepsilon_K + \sum_{k=1}^K \alpha_k\,\Delta_n^{(k)}$. Here, $\Delta_n^{(k)} \triangleq \max_{\mathbf{z}\in\mathcal{Z}} \mathbf{w}_k^\top \mathbf{f}(\mathbf{z}) - \max_{\mathbf{z}\in\mathcal{Z}_n} \mathbf{w}_k^\top \mathbf{f}(\mathbf{z})$ represents the scalarized optimality gap for the $k$-th direction.
\end{assumption}
Based on the above assumptions, the baseline cumulative hypervolume regret bound for a standard global GP-UCB approach is established as follows.

\begin{proposition}\label{prop:gp_G-MOBO_hv_rkhs}
With probability at least $1-\delta$, the cumulative hypervolume regret is globally bounded by:
\begin{equation}
\label{eq:hv_cum_regret_bound_rkhs}
R_N^{HV} \le N\,\varepsilon_K + C\left(\sum_{k=1}^K \alpha_k\right)\sqrt{N\,\beta_N\,\gamma_N(\mathcal{Z})},
\end{equation}
where $\gamma_N(\mathcal{Z})$ is the maximum information gain for the chosen GP kernel over the domain $\mathcal{Z}$, and $C>0$ is a constant dependent on the noise level and confidence scaling. The exact growth rate of $\gamma_N(\mathcal{Z})$ depends fundamentally on the kernel (\textit{e.g.}, $\mathcal{O}((\log N)^{d+1})$ for a Squared Exponential kernel) \cite{srinivas2010gaussian}. Here, $\beta_N$ is the GP-UCB confidence multiplier chosen such that the uniform confidence event holds for every scalarized function $s_k(\mathbf{z}) \triangleq \mathbf{w}_k^\top \mathbf{f}(\mathbf{z})$ with probability at least $1-\delta$.
\end{proposition}

\begin{proof}
According to Assumption~\ref{ass:function_regularity}, each objective function $f_m$ resides in the RKHS $\mathcal{H}_k$. By the inherent linearity of $\mathcal{H}_k$, the scalarized function $s_k(\mathbf{z}) = \sum_{m=1}^M w_{k,m} f_m(\mathbf{z})$ is also an element of $\mathcal{H}_k$. By the triangle inequality and the absolute homogeneity of the RKHS norm, its norm is strictly bounded by $\|s_k\|_{\mathcal{H}_k} \le \sum_{m=1}^M |w_{k,m}|\,B_m$. Consequently, applying the standard GP-UCB regret analysis \cite{srinivas2010gaussian} to each bounded scalar function $s_k$, the cumulative scalar regret $R_N^{(k)} \triangleq \sum_{n=1}^{N} ( \max_{\mathbf{z}\in\mathcal{Z}}s_k(\mathbf{z}) - s_k(\mathbf{z}_n) )$ satisfies $R_N^{(k)} \le C\sqrt{N\,\beta_N\,\gamma_N(\mathcal{Z})}$ with high probability. By definition, the scalar gap $\Delta_n^{(k)}$ is strictly upper-bounded by the instantaneous regret $\max_{\mathbf{z}\in\mathcal{Z}} s_k(\mathbf{z}) - s_k(\mathbf{z}_n)$. Summing over $n$ yields $\sum_{n=1}^N \Delta_n^{(k)} \le R_N^{(k)}$. Substituting this inequality into the hypervolume decomposition from Assumption~\ref{ass:hypervolume_decomposition}, the cumulative hypervolume regret is bounded by $R_N^{HV} = \sum_{n=1}^N (HV^\star - HV_n) \le \sum_{n=1}^N ( \varepsilon_K + \sum_{k=1}^K \alpha_k \Delta_n^{(k)} ) \le N \varepsilon_K + \sum_{k=1}^K \alpha_k R_N^{(k)}$. Inserting the scalar regret bound directly completes the proof.
\end{proof}

\subsection{Sample Efficiency Analysis}\label{sec:sample_efficiency}
\rev{While Proposition~\ref{prop:gp_G-MOBO_hv_rkhs} establishes a rigorous baseline regret bound for standard global GP-based optimization, the corresponding GP-based G-MOBO instantiation structurally accelerates convergence through two key mechanisms: localized TRs and grey-box modeling.} We first quantify the sample efficiency gained by the localized search strategy. Suppose the total evaluation budget $N$ is dynamically decomposed into $P$ phases. During each phase $p$, $N_p$ queries are evaluated strictly within a local TR $\mathcal{T}^{(p)} \subseteq \mathcal{Z}$, satisfying the budget constraint $\sum_{p=1}^P N_p = N$.

\begin{proposition}\label{prop:morbo_tr_sample_eff}
Under the localized TR strategy and assuming the GP-UCB confidence bound is valid within each TR phase, the cumulative hypervolume regret is globally bounded with probability at least $1-\delta$ by:
\begin{equation}
\label{eq:morbo_hv_bound}
R_N^{HV} \le N\varepsilon_K + C\left(\sum_{k=1}^K \alpha_k\right)\sqrt{\beta_N \sum_{p=1}^P N_p\,\gamma_{N_p}(\mathcal{T}^{(p)})},
\end{equation}
where $\gamma_{N_p}(\mathcal{T}^{(p)})$ denotes the maximum information gain restricted to the sub-domain $\mathcal{T}^{(p)}$. 
\end{proposition}

\begin{proof}
Based on the hypervolume decomposition, the global regret is bounded by the weighted sum of scalar regrets. Under the TR strategy, the total scalar regret $R_N^{(k)}$ is physically decomposed into the sum of regrets accumulated across the $P$ local phases. Applying the standard GP-UCB bound to each restricted domain $\mathcal{T}^{(p)}$ yields:
\begin{equation}
R_N^{(k)} = \sum_{p=1}^P R_{N_p}^{(k)}(\mathcal{T}^{(p)}) \le \sum_{p=1}^P C\sqrt{N_p\,\beta_N\,\gamma_{N_p}(\mathcal{T}^{(p)})}.
\end{equation}
To aggregate these localized bounds, the Cauchy-Schwarz inequality is applied, yielding:
\begin{equation}
\begin{split}
\sum_{p=1}^P \sqrt{N_p\,\gamma_{N_p}(\mathcal{T}^{(p)})} 
&\le \sqrt{\left(\sum_{p=1}^P N_p\right) \left(\sum_{p=1}^P \gamma_{N_p}(\mathcal{T}^{(p)})\right)} \\
&= \sqrt{N \sum_{p=1}^P \gamma_{N_p}(\mathcal{T}^{(p)})}.
\end{split}
\end{equation}
Substituting this localized aggregation back into the initial scalar regret expression, and subsequently into the global hypervolume regret bound, completes the proof.
\end{proof}

\noindent \textbf{Remark.} The bound derived in Proposition~\ref{prop:morbo_tr_sample_eff} explicitly highlights the mathematical advantage of the TR strategy. Because the information gain is monotonically dependent on the domain size, the condition $\mathcal{T}^{(p)} \subseteq \mathcal{Z}$ guarantees that $\gamma_{N_p}(\mathcal{T}^{(p)}) \le \gamma_{N_p}(\mathcal{Z})$. Consequently, the aggregated cumulative term $\sum_{p=1}^P \gamma_{N_p}(\mathcal{T}^{(p)})$ is systematically smaller than the global information gain $\gamma_N(\mathcal{Z})$. This strict reduction directly translates to a tighter regret bound and superior sample efficiency compared to unconstrained global exploration.

Beyond the topological advantages of reducing the search domain via localized TRs, the framework further accelerates convergence by exploiting the explicitly known analytical structure of the ISAC problem through grey-box surrogate modeling. Suppose each objective function admits an explicit grey-box decomposition $f_m(\mathbf{z}) = m_m(\mathbf{z}) + \Delta_m(\mathbf{z})$, where $m_m(\mathbf{z})$ is a deterministically known prior mean and $\Delta_m(\mathbf{z})$ is an unknown residual residing in $\mathcal{H}_k$ with a strictly smaller norm bound $\|\Delta_m\|_{\mathcal{H}_k} \le B_{\Delta,m}$, yielding a significantly reduced confidence parameter $\beta_N^{(\Delta)}$.

\begin{proposition}[Sample Efficiency of Grey-Box Modeling]\label{prop:greybox_sample_eff}
By applying the GP-UCB strategy exclusively to the residual function $\Delta_m$, the cumulative hypervolume regret is bounded with probability at least $1-\delta$ by:
\begin{equation}
\label{eq:greybox_hv_bound}
R_N^{HV} \le N\varepsilon_K + C\left(\sum_{k=1}^K \alpha_k\right)\sqrt{N\,\beta_N^{(\Delta)}\,\gamma_N(\mathcal{Z})}.
\end{equation}
\end{proposition}

\begin{proof}
Let the objective observation be $y_{m,n} = f_m(\mathbf{z}_n) + \epsilon_{m,n}$. Since the mean component $m_m(\mathbf{z})$ is explicitly known from the physical model, a residual observation is determined as:
\begin{equation}
\tilde{y}_{m,n} \triangleq y_{m,n} - m_m(\mathbf{z}_n) = \Delta_m(\mathbf{z}_n) + \epsilon_{m,n}.
\end{equation}
Consequently, learning the original function $f_m$ is statistically isomorphic to learning the zero-mean residual $\Delta_m$. 

For any scalarization direction $k$, the objective decomposes into a known deterministic part and an unknown residual $s_k^{(\Delta)}(\mathbf{z}) \triangleq \mathbf{w}_k^\top \boldsymbol{\Delta}(\mathbf{z})$. As established previously, this scalarized residual resides in the function space $\mathcal{H}_k$ with a norm tightly bounded by $\sum_{m=1}^M |w_{k,m}| B_{\Delta,m}$. Note that the required confidence parameter $\beta_N$ scales monotonically with the RKHS norm of the target function, applying the GP-UCB bound directly to the localized residual $s_k^{(\Delta)}$ yields:
\begin{equation}
R_N^{(k)} \le C\sqrt{N\,\beta_N^{(\Delta)}\,\gamma_N(\mathcal{Z})}.
\end{equation}
Substituting this bound into the unified hypervolume decomposition yields the exact result presented in \eqref{eq:greybox_hv_bound}.
\end{proof}

\rev{Under the stated GP/RKHS assumptions, this structural reduction yields a tighter cumulative regret bound for the GP-based grey-box realization than for its GP-based black-box counterpart. The following result quantifies how bounded RF predictive-moment approximation modifies this reference characterization.}

\rev{Random-forest predictions constructed from subsampled tree ensembles admit asymptotic distributional characterizations and estimable variance under suitable regularity conditions \cite{mentch2016quantifying}. Together with the bounded model-misspecification framework used in kernelized bandit analysis \cite{bogunovic2021misspecified}, this motivates explicitly separating the RF predictive-moment discrepancy from the GP/RKHS reference. Accordingly, we impose the following conditional uniform bound over each active TR.}

\begin{assumption}\label{ass:rf_moment_approx}
\rev{For the $k$-th nonnegative scalarization weight $\mathbf w_k$, define $s_k(\mathbf z)\triangleq\mathbf w_k^T\mathbf f(\mathbf z)=\sum_{m=1}^{M}w_{k,m}f_m(\mathbf z)$. For every scalarization direction $k$, iteration $n$, and configuration $\mathbf z$ in the active TR, let $\mu_n^{\rm GP}(\mathbf z)$ and $\sigma_n^{\rm GP}(\mathbf z)$ denote the GP reference moments of $s_k(\mathbf z)$, and let $\widehat\mu_n^{\rm RF}(\mathbf z)$ and $\widehat\sigma_n^{\rm RF}(\mathbf z)$ denote the corresponding moments induced by the RF realization. The scalarization index $k$ is suppressed in the moment notation for readability. Assume}
\begin{align}
\rev{\widehat\mu_n^{\rm RF}(\mathbf z)}
&\rev{=\mu_n^{\rm GP}(\mathbf z)+e_{\mu,n}(\mathbf z),}\label{eq:rf_gp_mean_approx}\\
\rev{\widehat\sigma_n^{\rm RF}(\mathbf z)}
&\rev{=\sigma_n^{\rm GP}(\mathbf z)+e_{\sigma,n}(\mathbf z),}\label{eq:rf_gp_std_approx}
\end{align}
\rev{where $|e_{\mu,n}(\mathbf z)|\leq\varepsilon_{\mu,n}$ and $|e_{\sigma,n}(\mathbf z)|\leq\varepsilon_{\sigma,n}$ uniformly over that region. Define the predictive-moment approximation penalty as}
\begin{equation}
\rev{\varepsilon_{{\rm app},n}\triangleq\varepsilon_{\mu,n}+\sqrt{\beta_n}\,\varepsilon_{\sigma,n}.}
\label{eq:rf_approx_penalty}
\end{equation}
\end{assumption}

\rev{On the GP confidence event $|s_k(\mathbf z)-\mu_n^{\rm GP}(\mathbf z)|\leq\sqrt{\beta_n}\sigma_n^{\rm GP}(\mathbf z)$, Assumption~\ref{ass:rf_moment_approx} and the triangle inequality give}
\begin{equation}
\rev{|s_k(\mathbf z)-\widehat\mu_n^{\rm RF}(\mathbf z)|
\leq \sqrt{\beta_n}\widehat\sigma_n^{\rm RF}(\mathbf z)
+\varepsilon_{{\rm app},n}.}
\label{eq:rf_approx_confidence}
\end{equation}

\begin{proposition}\label{prop:rf_approx_hv}
\rev{Suppose Assumptions~\ref{ass:function_regularity}, \ref{ass:hypervolume_decomposition}, and \ref{ass:rf_moment_approx} hold, and apply the same localized confidence-based scalarized selection argument as in Proposition~\ref{prop:morbo_tr_sample_eff}. Then, with probability at least $1-\delta$,}
\begin{equation}
\begin{split}
\rev{R_N^{\rm HV,RF}}
\rev{\leq {}}&\rev{N\varepsilon_K
+C\!\left(\sum_{k=1}^{K}\alpha_k\right)
\sqrt{\beta_N\sum_{p=1}^{P}N_p\gamma_{N_p}(\mathcal T^{(p)})}}\\
&\rev{{}+C_{\rm app}\sum_{n=1}^{N}\varepsilon_{{\rm app},n},}
\end{split}
\label{eq:rf_hv_regret_bound}
\end{equation}
\rev{where $C_{\rm app}>0$ absorbs the scalarization coefficients and confidence-analysis constants.}
\end{proposition}

\begin{proof}
\rev{Fix a scalarization direction $k$, and let $\mathbf z_{k,n}^{\star}$ maximize $s_k$ over the active TR at iteration $n$. On the confidence event in \eqref{eq:rf_approx_confidence}, the confidence-based selection rule and its upper and lower confidence bounds give the instantaneous scalar regret}
\begin{equation}
\rev{r_{k,n}^{\rm RF}
\triangleq s_k(\mathbf z_{k,n}^{\star})-s_k(\mathbf z_n)
\leq 2\sqrt{\beta_n}\,\widehat\sigma_n^{\rm RF}(\mathbf z_n)
+2\varepsilon_{{\rm app},n}.}
\label{eq:rf_instantaneous_regret}
\end{equation}
\rev{Assumption~\ref{ass:rf_moment_approx} further gives $\widehat\sigma_n^{\rm RF}\leq\sigma_n^{\rm GP}+\varepsilon_{\sigma,n}$ and $\sqrt{\beta_n}\varepsilon_{\sigma,n}\leq\varepsilon_{{\rm app},n}$. Thus, within each phase, summing \eqref{eq:rf_instantaneous_regret} and applying the Cauchy-Schwarz inequality together with the localized GP information-gain bound controls the accumulated predictive standard deviations. Summing over all phases and then applying Assumption~\ref{ass:hypervolume_decomposition} yields}
\begin{equation}
{\color{black}
\begin{aligned}
R_N^{(k),\rm RF}
&\leq C\sqrt{\beta_N\sum_{p=1}^{P}N_p
\gamma_{N_p}(\mathcal T^{(p)})}
+C_k^{\rm app}\sum_{n=1}^{N}\varepsilon_{{\rm app},n},\\[2pt]
R_N^{\rm HV,RF}
&\leq N\varepsilon_K
+C\!\left(\sum_{k=1}^{K}\alpha_k\right)
\sqrt{\beta_N\sum_{p=1}^{P}N_p
\gamma_{N_p}(\mathcal T^{(p)})}\\
&\quad+C_{\rm app}\sum_{n=1}^{N}\varepsilon_{{\rm app},n}.
\end{aligned}
}
\label{eq:rf_regret_proof_steps}
\end{equation}
\rev{Here, $C_k^{\rm app}$ collects the two confidence-width approximation contributions for direction $k$, and their scalarization-weighted sum is absorbed into $C_{\rm app}$. The last inequality is exactly \eqref{eq:rf_hv_regret_bound}, completing the proof.}
\end{proof}

\noindent\textbf{Remark.}
\rev{Proposition~\ref{prop:rf_approx_hv} does not assert that an arbitrary RF surrogate inherits the GP/RKHS regret guarantee. Instead, it quantifies the additional cost of replacing the GP predictive moments with their RF approximations. The first two terms in \eqref{eq:rf_hv_regret_bound} are inherited from the localized GP/RKHS analysis, while the last term accumulates the RF predictive-moment errors. Therefore, the RF implementation preserves the no-average-regret property only when $N^{-1}\sum_{n=1}^{N}\varepsilon_{{\rm app},n}\to 0$ and the localized GP/RKHS terms are sublinear. If this average approximation error does not vanish, Proposition~\ref{prop:rf_approx_hv} guarantees convergence only up to an error determined by the RF approximation accuracy.}

\subsection{\rev{Practical Operation and Computational Complexity}}
\rev{Each G-MOBO iteration consists of four operations: updating the surrogate, generating candidates inside the TR, evaluating the EHI acquisition function, and physically evaluating the selected FAS configuration. Let $m$ be the number of observed configurations, $M_{\rm RF}$ the number of RF trees, $D_{\rm tree}$ their average depth, $S_{\rm TR}$ the number of TR candidates, $N_{\rm mc}$ the number of Monte Carlo samples used for EHI, and $|\mathcal P_m|$ the current Pareto-set size. The RF training cost is approximately $\mathcal O(M_{\rm RF}m\log m)$, while predicting all TR candidates costs $\mathcal O(S_{\rm TR}M_{\rm RF}D_{\rm tree})$. Candidate generation is linear in the number of generated configurations and decision dimension. For the two-objective case, Monte Carlo EHI evaluation costs approximately $\mathcal O\!\left(S_{\rm TR}N_{\rm mc}(C_g+|\mathcal P_m|)\right)$, where $C_g$ is the cost of propagating one constituent sample through the known ISAC mapping. One selected configuration then incurs the physical evaluation cost $C_{\rm eval}$.} \rev{The feasible configuration space grows combinatorially with the number of candidate ports. G-MOBO avoids full enumeration by training only on the observed configurations and evaluating $S_{\rm TR}$ candidates within the active TR at each iteration.}

\begin{figure*}[t] 
    \centering
    % Subfigure (a): Total Utility vs Power
    % 瀵瑰簲鍥剧墖: mo_total_utility_vs_power_with_morbo
    \begin{subfigure}[b]{0.32\textwidth} 
        \centering
        \includegraphics[width=\linewidth]{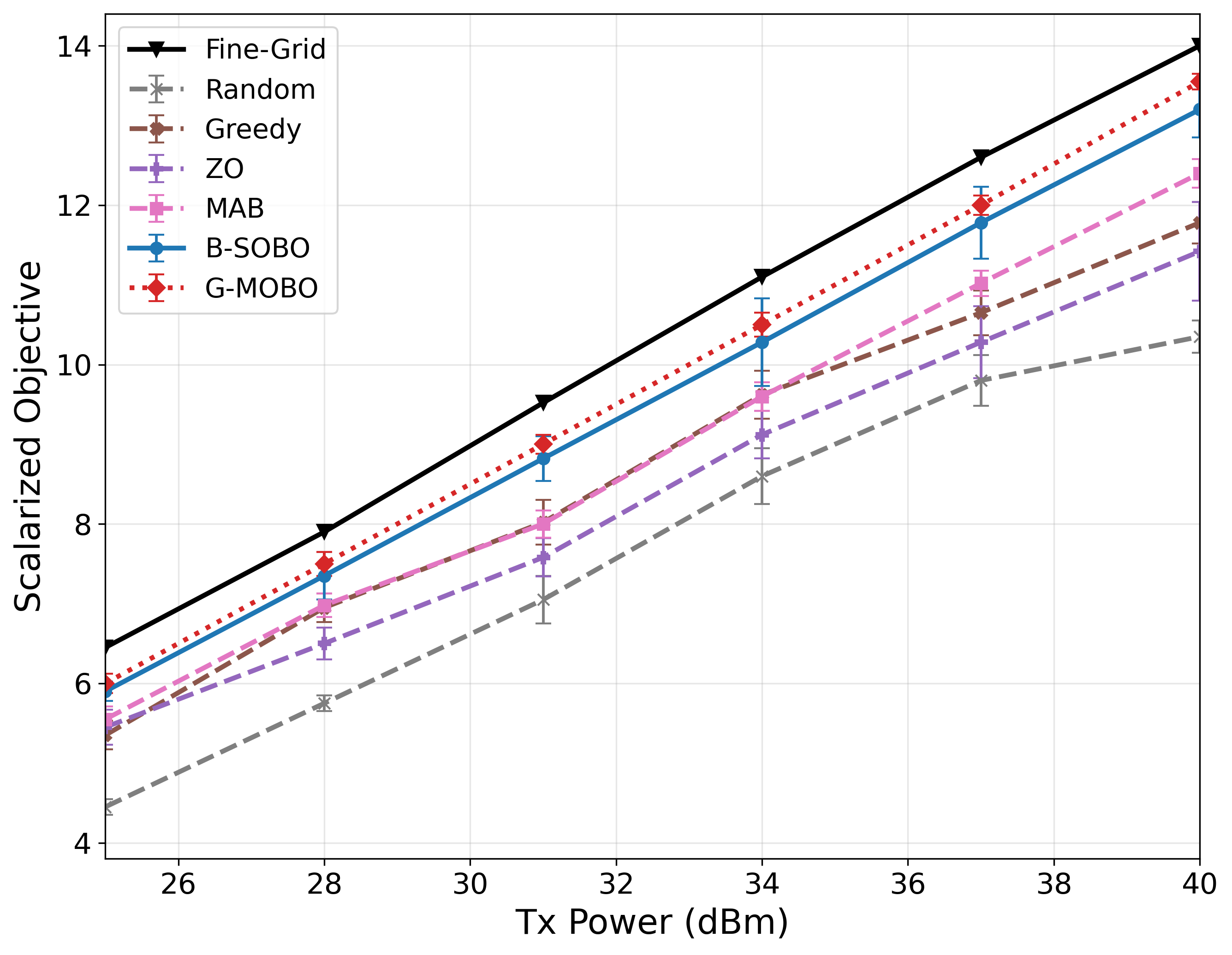}
        \caption{Scalarized Objective vs. Tx Power.}
        \label{fig:powersweep}
    \end{subfigure}
    \hfill 
    % Subfigure (b): Hypervolume Convergence
    % 瀵瑰簲鍥剧墖: ablation_surrogate_tr_hv
    \begin{subfigure}[b]{0.32\textwidth}
        \centering
        \includegraphics[width=\linewidth]{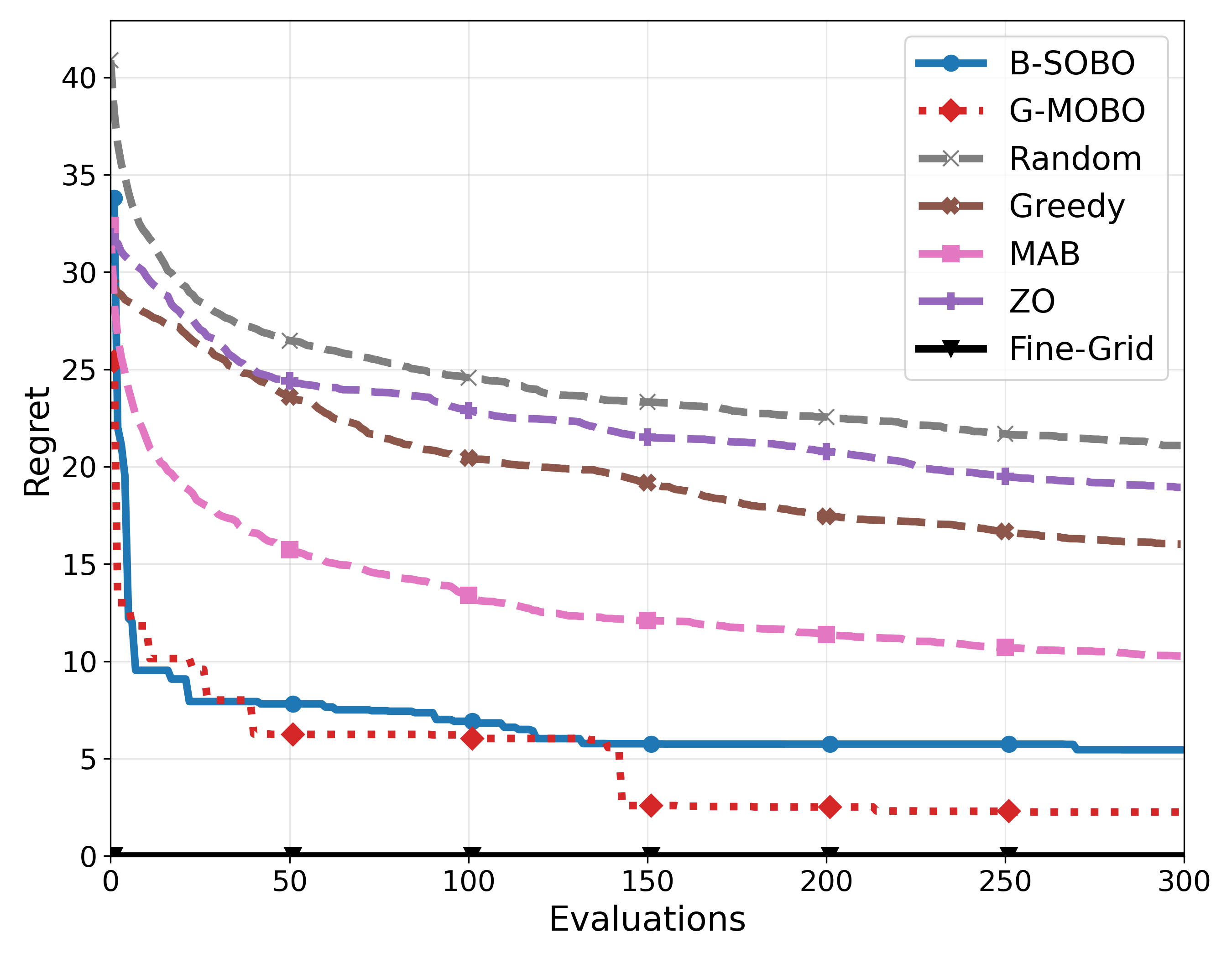}
        \caption{Hypervolume Convergence.}
        \label{fig:convergence}
    \end{subfigure}
    \hfill 
    % Subfigure (c): Pareto Front
    % 瀵瑰簲鍥剧墖: ablation_surrogate_tr_pareto
    \begin{subfigure}[b]{0.32\textwidth}
        \centering
        \includegraphics[width=\linewidth]{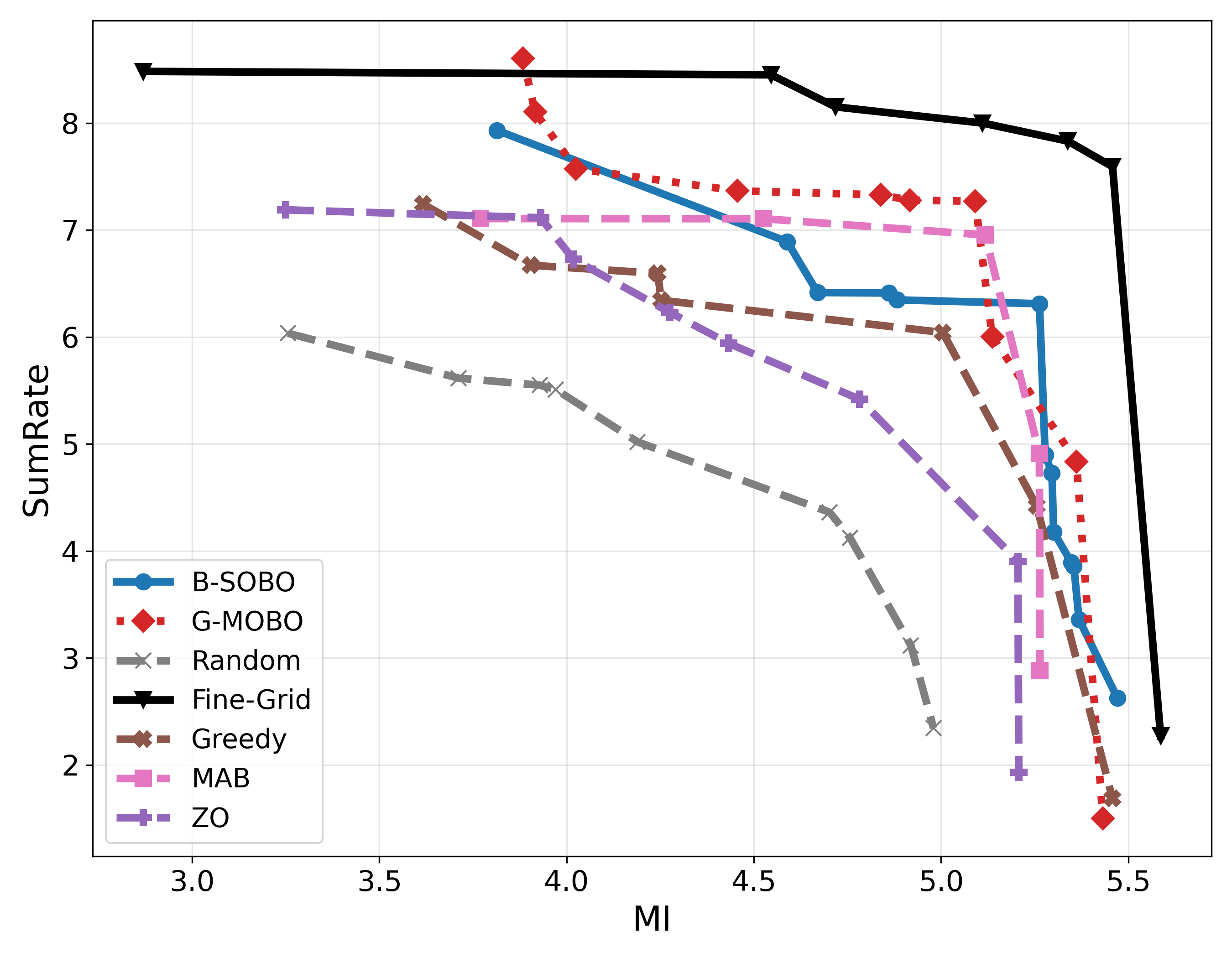}
        \caption{Sum-Rate vs. MI (Pareto Front).}
        \label{fig:tradeoff}
    \end{subfigure}
    
    \caption{Performance evaluation: (a) Impact of transmit power budget on Scalarized objective ($\alpha = 0.5$), (b) Hypervolume convergence over evaluations, and (c) The trade-off between communication (Sum-Rate) and sensing (MI).}
    \label{fig:combined_perf}
\end{figure*}

\section{Simulations}
\rev{This section compares the proposed G-MOBO framework with representative optimization and tracking baselines in terms of static objective performance, sample efficiency, Pareto-front quality, sensitivity to the FAS configuration and surrogate design, and dynamic tracking capability. The key simulation parameters are summarized in Table~\ref{tab:sim_params}, including the 28-GHz carrier frequency, network geometry and mobility settings, $4\lambda\times4\lambda$ FAS aperture with a $9\times9$ candidate-port grid, $N_t=N_r=4$ active fluid antennas, SI cancellation level, quantized phase and orientation codebooks, and the BO evaluation-budget and TR settings. The considered scenario comprises one BS, two vehicular users, two UAV sensing targets, and three static clutter sources. For each independent trial, the clutter locations are drawn uniformly over the monitored horizontal region, fixed throughout the trial, and converted into two-way geometric echo responses using the same steering construction as the target channel.}

\subsection{Static Performance}
To comprehensively evaluate \rev{the proposed method}, we compare it against \rev{the following benchmarks}:
% 鍗曠嫭鍒?
\begin{itemize}
    \item \rev{\textbf{Fine-Grid Search:} A budgeted multi-start search that scans the discrete orientation and precoding grids and performs coordinate-wise transmit/receive port swaps; it does not enumerate the full joint configuration space.}
    
    \item \textbf{Random Search:} Randomly samples configurations from the feasible space.
    \item \textbf{Greedy Search \cite{zhu2024multiuser}:} Iteratively selects the locally optimal configuration, utilized to demonstrate the necessity of long-term optimization planning.
    \item \textbf{MAB \cite{song2023two}:} An established derivative-free approach utilizing the Upper Confidence Bound (UCB) strategy for iterative action selection.
    \item \textbf{ZO Optimization \revcite{zeng2024csi}:} A state-of-the-art CSI-free method estimating mathematical gradients via random perturbations to perform projected gradient descent.
    \item \rev{\textbf{Black-box Single-Objective Bayesian Optimization (B-SOBO) \revcite{cheng2025bayesian}:} Optimizes a scalarized end-to-end ISAC utility without grey-box modeling, Pareto-front maintenance, or TR search.}
\end{itemize}

\begin{table}[t]
\centering
\caption{Simulation Parameters}
\label{tab:sim_params}
\scriptsize
\setlength{\tabcolsep}{2pt}
\renewcommand{\arraystretch}{1.05}
\begin{tabular}{@{}>{\raggedright\arraybackslash}p{0.20\columnwidth}>{\raggedright\arraybackslash}p{0.29\columnwidth}>{\raggedright\arraybackslash}p{0.45\columnwidth}@{}}
\toprule
{Category} & {Parameters} & {Value} \\
\midrule
General System & Carrier Frequency & $f_0$ = 28~GHz \\
 & \rev{Rician factor} & \rev{$K_{\rm R}=10$~dB} \\
 & \rev{NLoS paths} & \rev{$L_k=4$} \\
 & \rev{Noise variance} & \rev{$\sigma_{\rm c}^2=\sigma_{\rm s}^2=1$} \\
\midrule
Scenario Geometry & BS Position & $b_0 = [0,0,30]$~m \\
 & Users & $K = 2$ Vehicles \\
 & Targets & $Q = 2$ UAVs \\
 & Clutter & $Q_c = 3$ \\
 & \rev{SI model} & \rev{Spherical-wave Rician, $K_{\rm SI}=10$~dB} \\
 & \rev{SI cancellation} & \rev{$100$~dB} \\
 & \rev{Dynamic mobility} & \rev{$\Delta t=0.01$~s, $T=3$~s} \\
 & \rev{Speed ranges} & \rev{Users: $8$ to $12$; targets: $6$ to $10$~m/s} \\
\midrule
FAS Configurations & \rev{FAS aperture} & \rev{$4\lambda\times4\lambda$} \\
 & Candidate Ports & $P = 9 \times 9$ grid \\
 & Port Spacing & $d = \lambda/2$ \\
 & Active fluid antennas & $N_t=4$, $N_r=4$ \\
 & \rev{Phase codebook} & \rev{$b=3$ bits ($8$ levels)} \\
 & \rev{Orientation codebook} & \rev{$8$ azimuth $\times$ $8$ elevation levels} \\
\midrule
BO Configurations & Evaluation budget & $N= 300$ \\
 & Initial sample size & $N_0 = 64$ \\
 & TR parameters & $(\gamma_{i},\gamma_{d}) = (1.5, 0.5)$ \\
 & \rev{EHI samples} & \rev{$N_{\rm mc}=128$} \\
 & \rev{Independent runs} & \rev{$5$ static trials; $2$ dynamic seeds} \\
 & Success tolerance & $\tau_s = 1$ \\
 & Failure tolerance & $\tau_f = 3$ \\
 & \rev{Reference point} & \rev{$\mathbf r=[0,0]^T$} \\
\bottomrule
\end{tabular}
\end{table}

\begin{figure*}[t]
    \centering
    \begin{subfigure}[t]{0.32\textwidth}
        \centering
        \includegraphics[width=\linewidth]{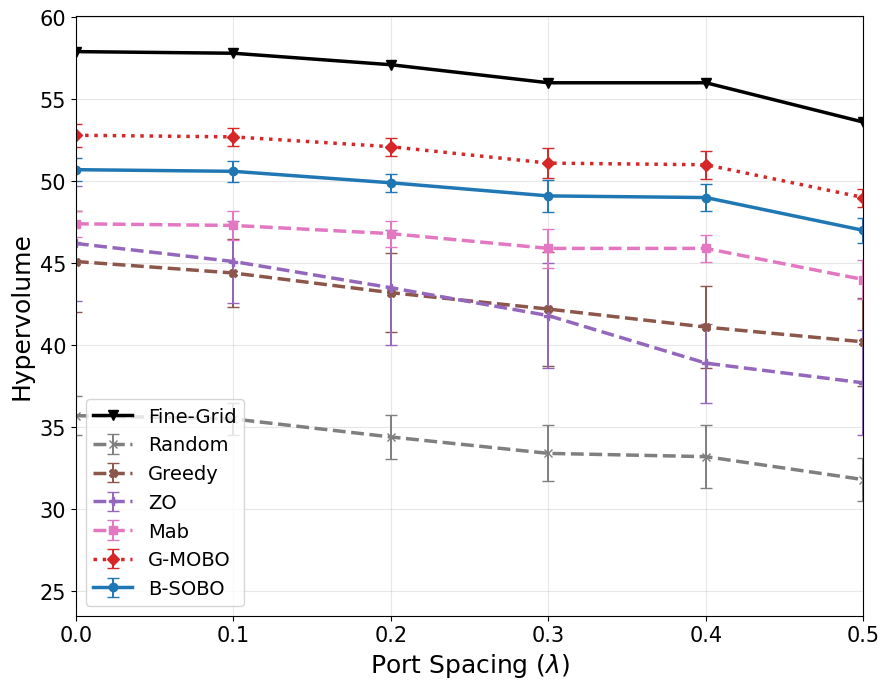}
        \caption{\rev{Port-spacing sweep.}}
        \label{fig:port_spacing}
    \end{subfigure}
    \hfill
    \begin{subfigure}[t]{0.32\textwidth}
        \centering
        \includegraphics[width=\linewidth]{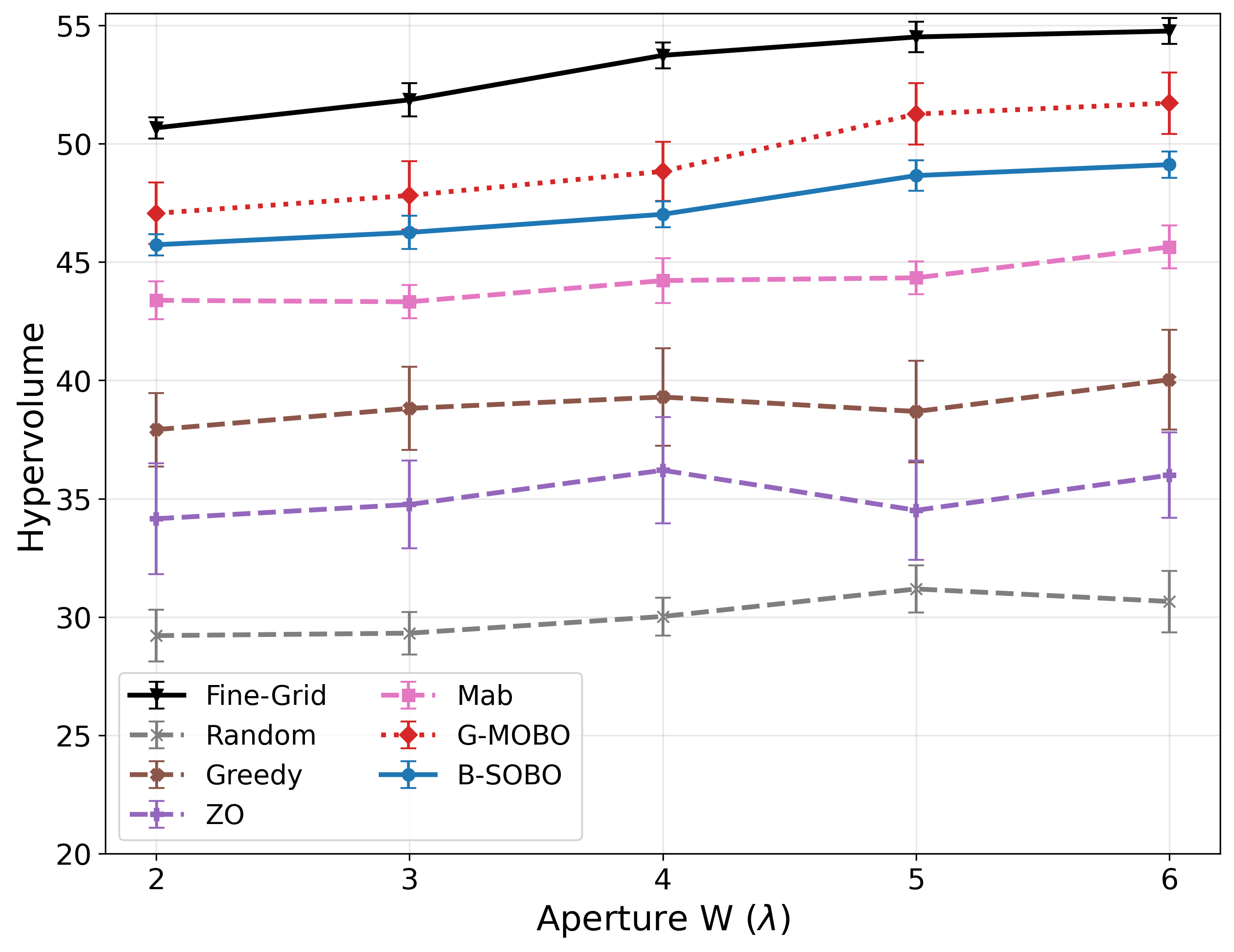}
        \caption{\rev{FAS-aperture sweep.}}
        \label{fig:aperture_sweep}
    \end{subfigure}
    \hfill
    \begin{subfigure}[t]{0.32\textwidth}
        \centering
        \includegraphics[width=\linewidth]{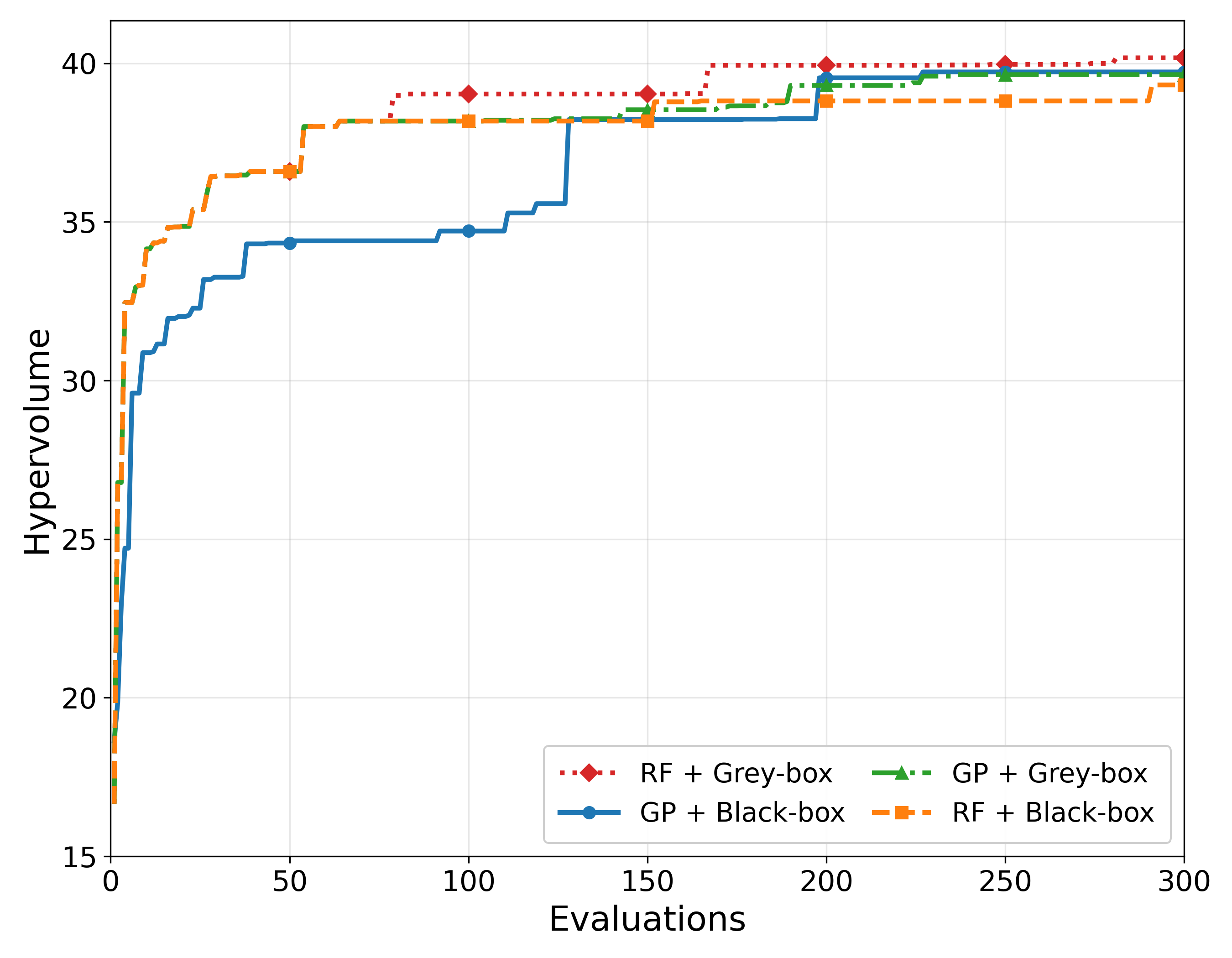}
        \caption{\rev{Internal ablation.}}
        \label{fig:gb_vs_bb}
    \end{subfigure}
    \caption{\rev{Sensitivity and ablation results: (a) achievable hypervolume versus normalized port spacing, (b) achievable hypervolume versus the side length $W$ of the square FAS aperture, and (c) hypervolume convergence under grey-box and black-box modeling with RF and GP surrogate realizations.}}
    \label{fig:sensitivity_ablation}
\end{figure*}

Fig.~\ref{fig:combined_perf} presents a comprehensive performance evaluation of the proposed G-MOBO framework against state-of-the-art baselines. Specifically, Fig.~\ref{fig:combined_perf}(a) investigates the optimization robustness under varying transmit power budgets $P_t$. While the scalarized objectives naturally improve with a higher signal-to-noise ratio (SNR), G-MOBO consistently achieves the highest utility, significantly outperforming heuristics like MAB and ZO in the high-power, interference-limited regime ($P_t = 40$ dBm). This widening performance gap highlights the superior interference management capability of the proposed method. Because configuration-dependent SI scales proportionally with the transmit power, heuristic methods lacking a global analytical model struggle to pinpoint precise beamforming and port combinations, whereas G-MOBO effectively models these spatial interactions to translate the additional power budget into net utility gains.

Fig.~\ref{fig:combined_perf}(b) analyzes the convergence of the hypervolume regret. G-MOBO demonstrates superior sample efficiency by rapidly minimizing the regret to approximately $2.5$ within 150 evaluations, \rev{closely approaching the Fine-Grid Search reference}. \rev{Furthermore, G-MOBO outperforms B-SOBO (regret plateau $\sim 5.5$) and escapes the local optima where MAB and Greedy stagnate (regret $\sim 10.0$ and $16.0$). This comparison isolates the benefit of combining multi-objective grey-box learning with localized Pareto-aware search, rather than attributing the gain to a particular surrogate alone.} The steep initial trajectory of the G-MOBO curve underscores the effectiveness of its localized TR strategy in accelerating early-stage exploration. Concurrently, \rev{its proximity to the Fine-Grid Search reference provides empirical evidence that} the variance-reduced grey-box surrogate successfully navigates the complex, multi-modal objective space without exhibiting premature convergence.

Finally, Fig.~\ref{fig:combined_perf}(c) visualizes the trade-off Pareto frontiers. The frontier achieved by G-MOBO strictly dominates all baselines. Particularly in the high-communication region (Sum-Rate $> 7$ bps/Hz), G-MOBO maintains an exceptional sensing MI above $5.0$ bps/Hz, whereas Greedy and ZO drop sharply below $3.8$ bps/Hz. Beyond simply pushing the frontier outward, G-MOBO successfully identifies a highly diverse set of non-dominated solutions distributed evenly across the entire trade-off spectrum. This confirms that explicitly maintaining a multi-objective acquisition strategy\rev{, coupled with a discrete-aware surrogate, prevents} the search from collapsing into localized extremes, thereby providing system designers with a rich continuum of flexible operating points for varying practical ISAC requirements.

\rev{Fig.~\ref{fig:port_spacing} illustrates the achievable hypervolume against the normalized port spacing $d/\lambda$. Across the evaluated spacings, G-MOBO provides higher search efficiency than B-SOBO, MAB, Greedy, and Random Search. Port spacing changes the smoothness and correlation structure of the discrete response and can therefore alter the relative approximation errors of different surrogates; however, the finite sweep is interpreted empirically rather than as a formal phase-transition theorem. As the port spacing becomes denser, the performance of ZO improves, narrowing the performance gap in the more tightly constrained configurations.}

\rev{To evaluate the practical aperture constraint independently of the nominal port density, Fig.~\ref{fig:aperture_sweep} varies the side length $W$ of the square FAS aperture from $2\lambda$ to $6\lambda$ while retaining the $\lambda/2$ candidate-port spacing and the same active transmit and receive antenna counts. The hypervolume achieved by G-MOBO increases from approximately $48.5$ to $51.5$ and remains higher than those of all practical baselines throughout the sweep. Most of the improvement is obtained by $W=4\lambda$, after which the gain becomes more gradual. This behavior indicates that enlarging the aperture initially provides useful spatial diversity for balancing communication and sensing, whereas the marginal benefit diminishes once the dominant spatial configurations can already be represented.}

\rev{Fig.~\ref{fig:gb_vs_bb} presents an internal ablation study comparing grey-box and black-box modeling across the RF and GP realizations. In the tested setting, the RF grey-box realization rapidly reaches a hypervolume plateau of approximately $40.2$, whereas its RF black-box counterpart remains near $35.0$ during early exploration before a later increase. The GP grey-box and black-box variants perform similarly and plateau near $39.5$. These results indicate that the value of grey-box decomposition depends on how accurately the chosen surrogate represents the discrete response: the implemented continuous-relaxation GP can incur kernel mismatch, whereas RF incurs a different partition-induced approximation error. The observation is specific to the evaluated setting and is not a universal RF-dominance claim.}

\subsection{Dynamic Extension}

\begin{figure}[t]
    \centering
    \includegraphics[width=\linewidth]{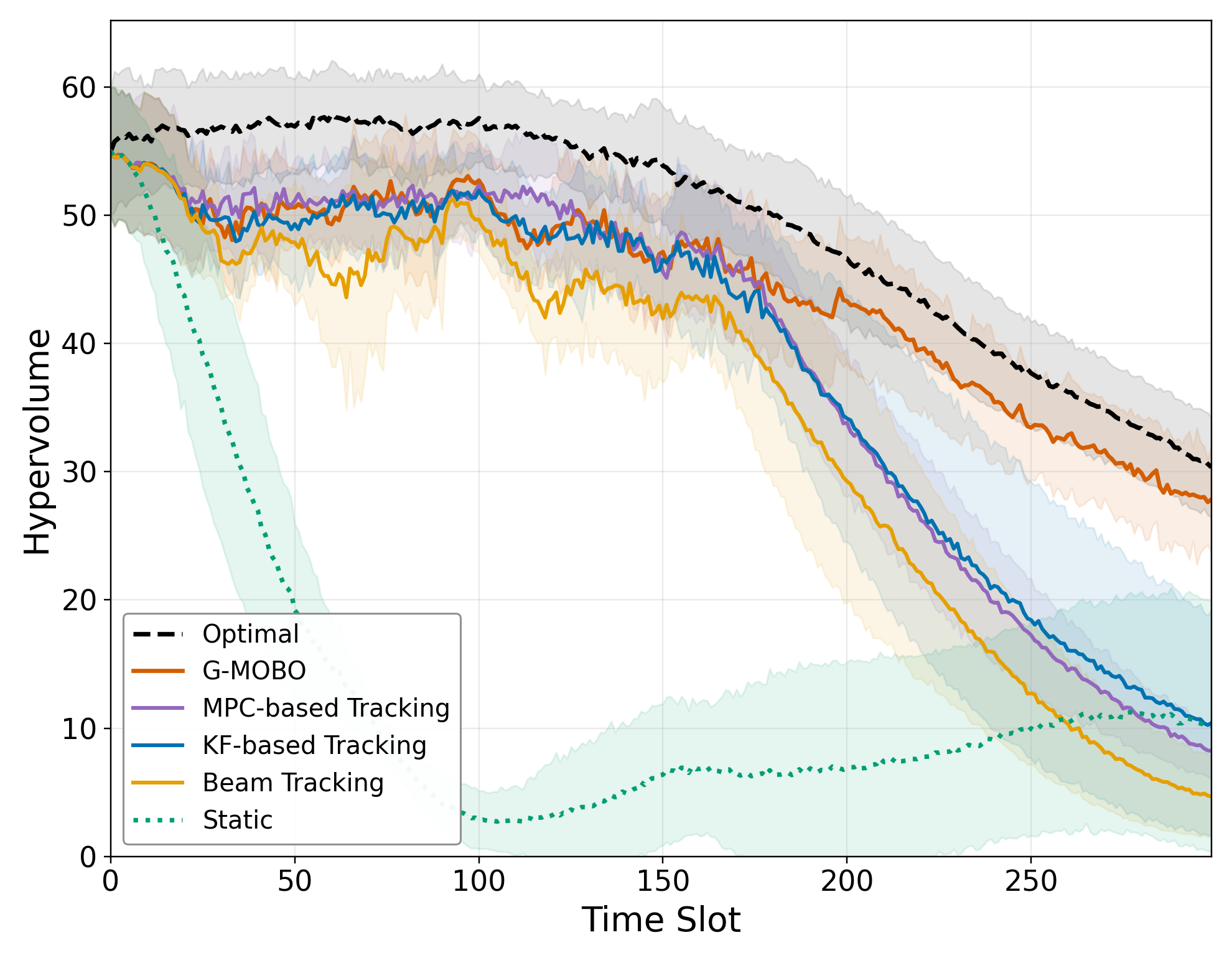}
    \caption{\rev{Dynamic tracking performance with Kalman filter (KF)-based, beam-tracking, and model predictive control (MPC)-based baselines. The shaded bands indicate across-run variability.}}
    \label{fig:dynamic_tracking}
\end{figure}
\rev{Fig.~\ref{fig:dynamic_tracking} evaluates the long-term performance stability of the proposed framework in a highly dynamic environment. The comparison includes the following benchmarks:}
\begin{itemize}
    \item \rev{\textbf{Optimal:} The multi-start greedy reference is rerun independently at every time slot. }
    \item \rev{\textbf{KF-based Tracking:} A KF predicts the next-slot channel/geometric state, after which the port and precoding configuration is updated using the predicted state \cite{li2025ekf}.}
    \item \rev{\textbf{Beam Tracking:} A predictive beam-tracking pipeline updates the beam direction and associated FAS configuration from the tracked angular/channel state \cite{liu2022learning,zhang2024predictive}.}
    \item \rev{\textbf{MPC-based Tracking:} Following the standard MPC principle \cite{rawlings2017mpc}, a receding-horizon controller predicts the state evolution over a short horizon, optimizes the corresponding configuration sequence, and applies its first decision.}
    \item \rev{\textbf{Static:} The configuration selected at $t=0$ remains fixed throughout the dynamic horizon.}
\end{itemize}

\rev{The Static baseline rapidly loses utility because the initial configuration becomes stale as the environment evolves. KF-based Tracking, Beam Tracking, and MPC-based Tracking mitigate this initial degradation, but all three prediction-driven methods decline sharply after the stronger environmental drift around slots 180 to 200. G-MOBO remains closest to the per-slot reference among the practical methods over most of the horizon and degrades more gradually: at slot 300, its hypervolume remains approximately 28, compared with roughly 8 to 10 for KF- and MPC-based Tracking and about 5 for Beam Tracking. This result does not indicate that prediction-based methods are generally inferior. Instead, it suggests that direct feedback-driven adaptation may be more suitable in the considered setting, where both state-prediction error and configuration-model mismatch affect the achieved communication-sensing utility.}

\FloatBarrier
\section{Conclusions}
This paper investigated the joint optimization of port selection and precoding for FAS-ISAC systems. We proposed a G-MOBO method to characterize the communication and sensing trade-off and address the challenges posed by incomplete CSI, configuration-dependent SI, and the high-dimensional decision space. \rev{Specifically, G-MOBO learns the unknown configuration-dependent constituents while retaining the known analytical mapping, employs EHI to explore the communication and sensing Pareto frontier, and uses an adaptive TR to restrict the search to promising local regions.} The method was further extended with temporal adaptation to track the drifting Pareto frontier in dynamic environments. \rev{The theoretical analysis characterized the cumulative hypervolume regret of the GP/RKHS reference realization and quantified the additional error introduced by bounded RF predictive-moment approximation.} Simulation results demonstrated the effectiveness and robustness of G-MOBO in both static optimization and dynamic tracking.

% Highlight the five references explicitly recommended by Reviewers 2 and 3.
\makeatletter
\let\standardbibitem\bibitem
\newcommand{\reviewerbibcolor}[1]{%
    \color{black}%
    \def\reviewerbibkey{#1}%
    \def\reviewerbibmatch{hao2025fluid}\ifx\reviewerbibkey\reviewerbibmatch\color{black}\fi
    \def\reviewerbibmatch{khalili2024advanced}\ifx\reviewerbibkey\reviewerbibmatch\color{black}\fi
    \def\reviewerbibmatch{khalili2026movable}\ifx\reviewerbibkey\reviewerbibmatch\color{black}\fi
    \def\reviewerbibmatch{zeng2024csi}\ifx\reviewerbibkey\reviewerbibmatch\color{black}\fi
    \def\reviewerbibmatch{cheng2025bayesian}\ifx\reviewerbibkey\reviewerbibmatch\color{black}\fi
}
\renewcommand{\bibitem}{\@ifnextchar[\reviewerbibitemwithlabel\reviewerbibitem}
\newcommand{\reviewerbibitem}[1]{\reviewerbibcolor{#1}\standardbibitem{#1}}
\def\reviewerbibitemwithlabel[#1]#2{\reviewerbibcolor{#2}\standardbibitem[#1]{#2}}
\makeatother
\bibliographystyle{IEEEtran}
\bibliography{refs}
\color{black}

\end{document}